	\documentclass[
		paper=letter,%
		11pt,american,pagesize,%
		DIV=12,headinclude=true,headlines=0,footinclude=true,footlines=-5,twoside=semi%
	]{scrartcl}
	\def\version{arxiv}
	\usepackage{hyperref}

\def\draftmode{false}

\usepackage[utf8]{inputenc}
\makeatletter
\usepackage[T1]{fontenc}
\usepackage{lmodern}
\usepackage{slantsc}

\usepackage{babel}
\newcommand{\ifarxiv}[2]{\ifthenelse{\equal{\version}{arxiv}}{#1}{#2}}
\newcommand{\ifmanuscript}[2]{\ifthenelse{\equal{\version}{manuscript}}{#1}{#2}}
\newcommand{\ifconf}[2]{\ifthenelse{\equal{\version}{conference}}{#1}{#2}}
\newcommand{\ifdraft}[2]{\ifthenelse{\equal{\draftmode}{true}}{#1}{#2}}
\newcommand\iflipics[2]{\ifthenelse{\equal{\version}{conference} \OR \equal{\version}{XXX}}{#1}{#2}}
\newcommand\ifkoma[2]{\ifthenelse{\equal{\version}{manuscript} \OR \equal{\version}{arxiv}}{#1}{#2}}

\usepackage{array,multicol}
\usepackage{amsmath,amsfonts,amssymb,mathtools}
\usepackage{mleftright}\mleftright 
\usepackage{relsize,ifthen,xspace,enumitem,booktabs,adjustbox,needspace,pbox}
\usepackage{graphicx}

\usepackage{wref}

\usepackage{clrscode3e}

\ifkoma{
	\usepackage{newfloat}
	\DeclareFloatingEnvironment[%
			name=Algorithm,%
			placement=thbp,%
			within=none,%
		]{algorithm}
}{}
\iflipics{
	\usepackage{newfloat}
	\DeclareFloatingEnvironment[%
			name=Algorithm,%
			placement=thb,%
		]{algorithm}
}{}

\vfuzz \hfuzz	

\allowdisplaybreaks[3]

\usepackage{lscape} 

\newcommand\plaincenter[1]{%
	\mbox{}\hfill#1\hfill\mbox{}%
}

\usepackage{tikz}
\usepackage{relsize}
\usetikzlibrary{positioning,arrows.meta,shapes}
\usetikzlibrary{backgrounds,calc,trees,hobby}
\usetikzlibrary{decorations,decorations.pathmorphing,decorations.pathreplacing}

\pgfdeclarelayer{background}
\pgfsetlayers{background,main}

\AtBeginDocument{%
	\ifthenelse{\isundefined{\myshorttitle}}{
		\let\mytitle\@title%
	}{
		\let\mytitle\myshorttitle
	}
}

\ifkoma{
	\setlength\parindent{1.5em}
	\usepackage{colonequals}
	\usepackage[headsepline]{scrlayer-scrpage}
	\pagestyle{scrheadings}
	\clearscrheadfoot
	\AtBeginDocument{%
		\manualmark%
		\markleft{\mytitle}%
		\automark*[section]{}%
	}
	\ohead{\pagemark}
	\ihead{\headmark}
	\addtokomafont{caption}{\sffamily\small}
	\addtokomafont{captionlabel}{\sffamily\textbf}
	\setcapmargin{2em}
}{}

\hypersetup{
	final,
	unicode=true, 
	bookmarks=true,
	bookmarksnumbered=true,
	bookmarksdepth=2,
	bookmarksopen=true,
	breaklinks=true,
	hidelinks,
	pdfhighlight={/P}
}

\iflipics{
	\newtheorem{fact}[theorem]{Fact}
}{
	\usepackage[amsmath,hyperref,thmmarks]{ntheorem}
	
	\theorembodyfont{\slshape}
	\theoremseparator{:}
	\newtheoremstyle{proofstyle}%
	  {\item[\theorem@headerfont\hskip\labelsep ##1\theorem@separator]}%
	  {\item[\theorem@headerfont\hskip\labelsep ##1 of ##3\theorem@separator]}
	
	\theorempreskip{\topsep} 

	\newtheorem{theorem}{Theorem}[section]
	
	\theoremstyle{plain}
	\theorempreskip{\topsep}

	\newtheorem{lemma}[theorem]{Lemma}
	
	\newtheorem{corollary}[theorem]{Corollary}
	\newtheorem{definition}[theorem]{Definition}
	
	\theoremstyle{plain}
	\theorembodyfont{\upshape}

	\theoremsymbol{\raisebox{-.25ex}{$\Box$}}
	\qedsymbol{\raisebox{-.25ex}{$\Box$}}

	\theoremstyle{proofstyle}
	\newtheorem{proof}{Proof}
}{}

\newcommand*\ie{\mbox{i.\hspace{.2ex}e.}}
\newcommand*\eg{\mbox{e.\hspace{.2ex}g.}}

\newcommand*\Withoutlossofgenerality{\mbox{W.\hspace{.23ex}l.\hspace{.2ex}o.\hspace{.17ex}g.}\xspace}

\newcommand\Oh{O}
\def\.{\mskip1mu}

\newcommand{\surroundedmath}[3]{
	\mathchoice{
		#1{#2{#3}#2}%
	}{
		#1{#3}%
	}{
		#1{#3}%
	}{
		#1{#3}%
	}%
}

\newcommand\wwrel[1]{\surroundedmath{\mathrel}{\;\;}{#1}}

\newcommand{\relwithtext}[3][c]{%
	\mathrel{\underset{\mathclap{\makebox[\widthof{$=$}][#1]{\scriptsize#2}}}{#3}}%
}

\newcommand{\ESymbol}{\mathbb{E}}

\newcommand{\ProbSymbol}{\ensuremath{\mathbb{P}}}

\DeclarePairedDelimiterXPP\Prob[1]{\ProbSymbol}[]{}{%
	#1%
}
\DeclarePairedDelimiterXPP\E[1]{\ESymbol}[]{}{%
	#1%
}
\DeclarePairedDelimiterXPP\Eover[2]{\ESymbol_{#1}}[]{}{%
	#2%
}
\providecommand{\Prob}{} 
\providecommand{\E}{} 
\providecommand{\Eover}{} 

\makeatletter
\let\oldalign\align
\let\endoldalign\endalign

\newcommand*\numberthis[1][]{\stepcounter{equation}\tag{\theequation}}

\iflipics{}{
	\renewenvironment{align}{%
		\begingroup%
		\let\oldhalign\halign
		\def\halign{%
			\let\oldbreak\\%
			\def\nonnumberbreak{\nonumber\oldbreak*}%
			\def\\{%
				\@ifstar{\nonnumberbreak}{\oldbreak}%
			}%
			\oldhalign%
		}
		\oldalign%
	}{%
		\endoldalign%
		\endgroup%
	}
}

\newcommand\splitaftercomma[1]{%
  \begingroup
  \begingroup\lccode`~=`, \lowercase{\endgroup
    \edef~{\mathchar\the\mathcode`, \penalty0 \noexpand\hspace{0pt plus .25em}}%
  }\mathcode`,="8000 #1%
  \endgroup
}

\def\mydots{\xleaders\hbox to.5em{\hfill.\hfill}\hfill}

\newlength\tmpLenNotations

\allowdisplaybreaks[3]

\iflipics{
	\let\oldparagraph\paragraph
	\renewcommand\paragraph[1]{%
		\oldparagraph*{#1}
	}
}{
	\let\oldparagraph\paragraph
	\renewcommand\paragraph[1]{%
		\oldparagraph{#1.}
	}
}

\usepackage[bibtex]{url-doi-arxiv}

\let\oldthebibliography\thebibliography
\renewcommand\thebibliography[1]{%
	\oldthebibliography{#1}%
	\pdfbookmark[1]{References}{}%
}

\ifkoma{
	
}{}

\ifdraft{%
	\overfullrule=6mm
	
	\usepackage[color,notref,notcite]{showkeys}
	\definecolor{refkey}{gray}{.99}
	\colorlet{labelkey}{green!60!black!60}
	
	\usepackage[inline,nolabel]{showlabels}

	\showlabels{cite}
	\showlabels{citealt}
	\showlabels{citealp}
	\showlabels{citet}
	\showlabels{Citet}
	\showlabels{citep}
	\showlabels{citeauthor}
	\showlabels{Citeauthor}
	\showlabels{citefullauthor}
	\showlabels{citeyear}
	\showlabels{citeyearpar}
	\showlabels{wref}
	\showlabels{wpref}
	\showlabels{wtpref}
	\showlabels{wildref}
	\showlabels{wildpageref}
	\showlabels{wildtpageref}
}{}

\ifdraft{}{%
	\usepackage{microtype}
}

\let\epsilon\varepsilon

\newsavebox\tmpbox

\ifmanuscript{
	
}{}

\makeatother

\ifarxiv{
	\iflipics{
		\hideLIPIcs
		\nolinenumbers
	}{}
}{}

\usepackage{hyperref}

\hypersetup{hidelinks}

\ifkoma{
	\title{%
		Dynamic Optimality Refuted~-- \\\protect For Tournament Heaps
	}
	\def\myshorttitle{Dynamic Optimality Refuted~-- For Tournament Heaps}
	
	\author{
		J.\ Ian Munro%
		\footnote{%
			University of Waterloo, Canada, \texttt{imunro\,@\,uwaterloo.ca}
		}
	\and
		Richard Peng%
		\footnote{%
			Georgia Institute of Technology / Microsoft Research Redmond, USA, \texttt{richard.peng\,@\,gmail.com},
			part of this work was done while at the University of Waterloo.
		}
	\and
		Sebastian Wild%
		\footnote{%
			University of Waterloo, Canada, \texttt{wild\,@\,uwaterloo.ca}
		}
	\and
		Lingyi Zhang%
		\footnote{%
			University of Waterloo, Canada, \texttt{l367zhang\,@\,edu.uwaterloo.ca}
		}
	}
	\date{\small\today}
}{}

\tikzset{
	every node/.append style={inner sep=1pt},
	big treenode/.style={circle,draw,minimum size=35pt,fill=black!10},
	treenode/.style={circle,draw,minimum size=30pt,fill=black!10},
	empty/.style = {minimum size = 20pt},
	inner/.style={circle,draw,minimum size=20pt,fill=black!10},
	einner/.style={circle,draw,minimum size=10pt},
	ileaf/.style={rectangle,draw,minimum size=10pt},
	leaf/.style={rectangle,draw,minimum size=15pt,fill=black!10},
	subtree/.style={
		draw,
		isosceles triangle,
		shape border rotate=90,
		minimum width=20pt
	},
}

\begin{document}

\maketitle

\vspace{-1ex}

\begin{abstract}%
We prove a separation between offline and online algorithms
for finger-based tournament heaps undergoing key modifications.
These heaps are implemented by binary trees with keys stored
on leaves, and intermediate nodes tracking the min of
their respective subtrees.
They represent a natural starting point for studying self-adjusting
heaps due to the need to access the root-to-leaf path
upon modifications.
We combine previous studies on the competitive ratios of
unordered binary search trees by [Fredman WADS2011]
and on order-by-next request
by [Martínez-Roura TCS2000] and [Munro ESA2000] to show that
for any number of fingers, tournament heaps cannot handle
a sequence of modify-key operations with competitive ratio in $o(\sqrt{\log{n}})$.

Critical to this analysis is the characterization of the modifications
that a heap can undergo upon an access.
There are $\exp(\Theta(n \log{n}))$ valid heaps on $n$ keys, but
only $\exp(\Theta(n))$ binary search trees.
We parameterize the modification power
through the well-studied concept of fingers:
additional pointers the data structure can manipulate arbitrarily.
Here we demonstrate that fingers can be significantly more powerful
than servers moving on a static tree by
showing that access to $k$ fingers allow an offline algorithm to handle 
any access sequence with amortized cost
$O(\log_{k}(n) + 2^{\lg^{*}n})$.
\end{abstract}

\section{Introduction}
\label{sec:intro}

One of the most intriguing open questions in data structures is the 
\textit{dynamic-optimality conjecture}.
The conjecture states that splay trees can serve any sequence of operations with
at most a constant times the cost of the best (adaptive) binary-search-tree (BST) based method,
even if we allow the latter to know the sequence of accesses in advance 
(\ie, to work in ``offline'' mode).
Despite decades of active research and deep results~\cite{Wilber1989,%
	Munro2000,DemaineHarmonIaconoPatrascu2007,WangDS06,Iacono2005,%
	Harmon06:thesis,DemaineHarmonIaconoKanePatrascu2009,ChalermsookGKMS15,Iacono2016a,KozmaSaranurak2018,LevyTarjan2019}
the main conjecture remains wide open, and so is the more general question:
\begin{center}\slshape
	Is there an online binary-search-tree algorithm that~-- on any access sequence~--\\
	performs within a constant factor of the offline optimal for that sequence?
\end{center}

In this paper, we ask the same question for heaps%
\footnote{%
	In the context of this paper, by a ``heap'' we mean any tree-based priority-queue data structure.%
}.
Dynamic optimality of heaps has attracted a lot of interest recently
due to the work of Kozma and Saranurak~\cite{KozmaSaranurak2018} 
who formalized a correspondence between self-adjusting BSTs and 
self-adjusting heaps like pairing heaps.
They show that every heap algorithm in their model, ``stable heaps in sorting mode''
(discussed in \wref{sec:stable-heaps} in more detail),
implies a corresponding BST algorithm
with the same cost (up to constant factors and on the time-space inverted input).
If the converse holds, too, is unclear,
and they had to leave the question of dynamic optimality for heaps open.
We note that \emph{refuting} dynamic optimality for stable heaps 
would hence not have immediate consequences for the existence (or nonexistence) 
of dynamically optimal BSTs.

While the dynamic-optimality conjecture has spurred the much wider
study of online algorithms, 
\eg, \cite{ManasseMS90,BorodinE05:book,Hazan16:book,BansalBMN15,BubeckCLLM18,Lee2018},
competitiveness results are notoriously sensitive to
details of the model of computation.
In fact, the historical starting point of competitive analysis,
searching on linear lists~\cite{SleatorTarjan1985}, is taught in
graduate courses, but~-- as often overlooked~-- a more realistic
model allowing arbitrary rearrangements (which can be simplified
to linked-list operations) on the visited prefix,
allows to serve any sequence known in advance with $O(n \log{n})$ operations~--
significantly less than the $\Omega(n^2)$ lower bound for
even processing random online access sequences~\cite{MartinezRoura2000,Munro2000}.

For the binary-search-tree problem itself,
such a possibility, while strange, could also be consistent
with observed gaps between performances of splay algorithms
and more tuned dictionary data structures%
\footnote{%
	The performance
	gaps are often quoted as between $1.5 \times$ to $3 \times$,
	while the value of $\lg\lg{n}$ for most values of $n$
	in practice is at most $5$.
	To our knowledge, the $O(\log\log{n})$-competitive search tree
	algorithms~\cite{DemaineHarmonIaconoPatrascu2007,WangDS06}
	have not been evaluated in practice.%
}.
And as we will show in this paper,
the separation of online and offline algorithms
is a fact for heaps:
we refute dynamic optimality for heaps
based on tournament trees with decrease-key operations.

Our model of computation is a natural restriction of general pointer machines;
analogous to how computation on linked lists~\cite{MartinezRoura2000,Munro2000} 
and binary search trees~\cite{Harmon06:thesis,DemaineHarmonIaconoKanePatrascu2009}
have been defined.
Since nodes in a pointer machine have constant size,
arbitrary-degree heap-ordered forests (as used in stable heaps) 
are not a convenient choice as primitive objects of manipulation.
Instead, we use tournament trees~\cite[\S5.2.3]{Knuth98a:book}:
Here each node has two children, and all the original keys are
stored in the leaves.
Moreover, each internal node stores the minimum of its children.

The priority-queue operations can be implemented as follows:
The global minimum is always found at the root.
To extract the minimum, we follow the path of copies of the minimum to the leaf that stores it,
remove it (and its parent internal node) from the tree,
and update the labels on the path.
Insertion of a new element can be achieved by adding the new leaf and the old root
as the two children of a new root node; merging two queues is similar.
Assuming a pointer is provided to a leaf, we can also change its key to a different value
and update the labels on the path from this leaf to the root.

In this work, we focus on the propagation part of the operations,
particular after changing a key of a leaf, and we will assume the worst case, 
namely that propagation of label changes always continues all the way up to the root.
This corresponds to a sequence of decrease-key operations where the new key is the new minimum.%
\footnote{%
	Sequences of only change-key operations naturally occur, \eg, in merging $n$ runs;
	updates that \emph{all} produce a new minimum are much less natural, but
	are sufficient for our negative result. A dynamically optimal tournament heap would
	in particular have to handle such sequences optimally.
}
Our formal model encodes this
implicitly by requiring any access
to touch the root-to-leaf path:
\begin{definition}[Tournament trees]
\label{def:model}
	In the \textit{tournament-tree-with-$k$-fingers model of computation,}
	one maintains
	a collection of $n$ elements in the \textit{leaves} of a binary tree.
	To serve an access to $x$, we start with all fingers $F_1,\ldots,F_k$ pointing at the root.
	We can then use the following operations for each of the fingers $F_i$, $i=1,\ldots,k$:
	\begin{enumerate}[itemsep=0ex]
		\item Move $F_i$ to the parent, the left child or right child
			of its current location (provided the followed pointer is not null).
		\item Copy the location of $F_i$ into $F_0$ (the temporary finger).
		\item Move $F_i$ to the location $F_0$.
		\item Swap the subtree with root $F_0$ with the left or right child
			of the node at $F_i$'s current location.
		\item Detach the subtree with root $F_0$ from its parent and make it the left or right child
			of the node at $F_i$'s current location (provided the replaced pointer is null).
		\item Serve request from $F_i$ (provided $x$ is stored at its current location).
	\end{enumerate}
	Each sequence of operations must eventually serve the access (via the last operation).
	The cost of this access is taken to be the number of operations (total of all fingers).
\end{definition}
Our result for this model is the following separation of offline and online performance.
\begin{theorem}[Online/Offline Separation]
\label{thm:Main}
	For any value $k = k(n)$,
	the competitive ratio of tournament heaps with $k$-fingers is
	$\Omega(\max\{\log_k (n), \log k \}) = \Omega(\sqrt{\log n})$.
\end{theorem}
Moreover, while any online algorithm incurs amortized cost
$\Omega(\log n)$ per access for most
inputs even when $k$ is as large as $\sqrt{n}$,
we show that we can do much better even with a subpolynomial number of
fingers in the offline
case by present a simple, efficiently-computable offline algorithm.
\begin{theorem}[Efficient Offline Algorithm]
\label{thm:offline-algo}
	Any sequence of operations on a tournament heaps can be served using $k$ fingers
	with amortized $O(\log_{k}(n) + 2^{\lg^*(n)})$ cost per access.
\end{theorem}
\wref[Theorems]{thm:Main} and~\ref{thm:offline-algo} show the fundamental importance 
that the underlying rearrangement primitives play,
in stark contrast to the BST model 
where any subtree replacement can be simulated via rotations in linear time.
The reason for this difference comes mainly from the fact that there are 
only exponentially ($2^{\Theta(n)}$) many BSTs on $n$ keys,
but factorially ($2^{\Theta(n \log n)}$) many heaps%
\footnote{%
	This statement being true for both standard heap-ordered trees and tournament trees, 
	but \emph{not} for Kozma and Saranurak's stable heaps!%
},
so tree rearrangements are much more powerful in heaps,
and standard local primitives like rotations are no longer sufficient.
Recall how weak primitive operations were also the critique of the model 
for self-adjusting linear lists~\cite{MartinezRoura2000,Munro2000} mentioned earlier.
Adding the number of fingers as a parameter to the model allows us to precisely quantify
the effect of more powerful rearrangement operations in tournament heaps.

Note that tournament trees place no restrictions on the order of keys in the leaves;
they are thus essentially equivalent to leaf-oriented \emph{unordered} binary trees.
Fredman~\cite{Fredman2011,Fredman2012} previously studied a similar model,
namely unordered binary trees, where keys are stored in all nodes, 
but no ordering constraint is placed on them.
He only considers a single finger, and purely local restructuring (rotations and subtree swaps).
He proved that in his model, for any given online method, 
one can construct (adaptively and in exponential time) an
adversarial input on which this online method incurs cost $\Omega(n \log n)$,
whereas the same input can be also served with linear costs 
(by an offline method tailored to this purpose).

Our result strengthens Fredman's work by explicitly studying leaf-oriented trees 
and (more importantly) by taking more realistic rearrangement primitives into account.
The result is a much stronger
separation between offline and online
algorithms as soon as a non-constant number of fingers is available:
The worst-case cost of our offline algorithm on sufficiently long sequences is asymptotically
smaller than the average cost any online algorithm can achieve on random sequences.
With $k=\omega(1)$, we refute dynamic optimality for tournament heaps 
even for the average competitive ratio.

\paragraph{Outline}
In \wref{sec:background}, we discuss related work,
and introduce notations for describing our models of trees.
In \wref{sec:models}, we prove the (worst-case) separation
between online and offline methods.
\wref{sec:loglog} presents our efficient offline algorithm for many fingers.
\wref{sec:conclusion} summarizes our findings and lists open problems.

\section{Background}
\label{sec:background}

In this section we introduce some notation for the search tree model,
and summarize related works, in particular we discuss differences 
of related models for trees studied in the context of dynamic optimality.

Our analysis will use the standard big-O notation,
and we write $f\sim g$ to denote $f = g(1\pm o(1))$
We will use $\lg$ to denote the binary logarithm,
and $\lg^{(k)}$ to denote the $k$-times iterated logarithms, \ie,
$\lg^{(1)}(n) = \lg(n)$ and $\lg^{(k+1)}(n) = \lg(\lg^{(k)}(n))$.
By $\lg^*(n)$ we denote the smallest $k$ so that $\lg^{(k)}(n) \le 1$.
Intervals over integers ranging from $a$ to $b$ will be denoted
using $[a \ldots b]$, and $[a]$ will be the shorthand for
$[1 \ldots a]$.

Dynamic optimality asks whether a data structure that
sees a sequence of operations online, that is, one at a time,
can perform as well as a data structure that sees the entire
access sequence ahead of time.
The critical definition for studying dynamic optimality is
the definition of an access sequence.
We will denote the $n$ keys as $1 \ldots n$, and denote an
access sequence of length $m$ as
\(
A = a_1,\ldots, a_m.
\)
Usually, we are interested in the case $m \ge n$.

For such an access sequence,
the cost of an online algorithm
is the cost of it accessing $a_1$, and then $a_2$
and so on, while the offline optimal cost of accessing $A$
is the minimum total cost of accessing the entire sequence.
The competitiveness ratio of an online algorithm $\textsc{Alg}$
on inputs of size $n$ is then
\[
\lim_{m \rightarrow \infty}
\max_{\text{$A$ of length $m$}}
\frac{\text{cost}\left( \textsc{Alg}\left( A \right) \right)}
{\textsc{OfflineOptimum}\left( A \right)}.
\]
We say that $\textsc{Alg}$ is $f(n)$-competitive if its
competitive ratio on any input of size $n$ is at most~$f(n)$ for large enough $n$.
For the converse, to show that $\textsc{Alg}$ is \emph{not} $f(n)$-competitive, 
it suffices to find a specific family of (arbitrarily large) inputs 
where the competitive ratio is worse than $f(n)$.

Small differences in the models of computation
have profound consequences for the performances of heaps.
We therefore start by formally defining and comparing these models.

\subsection{Trees in the Pointer Machine Model, and Fingers}

All of our data structures will be modeled using pointer machines.
Here nodes of trees are represented as a collection of $O(1)$ pointers,
each pointing to some other node.
In particular, a binary tree is a collection of nodes
each pointing to a parent, and a left/right child;
some of these pointers can be ``null''.

Access and modifications of pointer-based data structures are done
by manipulating the pointers.
For this purpose, it is useful to consider ``fingers'', which
are special global pointers kept by the data structure at the
topmost level.
These objects can be viewed as generalizations of the ``root'' vertex,
which in this terminology is a static finger where all
subsequent accesses start from.
We keep the number of fingers as a parameter, $k$,
which is allowed to depend on the size $n$ of the data structure,
(similar to how the word size $w$ of a word-RAM may depend on $n$).
As mentioned in \wref{def:model}, the cost of
performing a sequence of accesses is entirely the number of
operations performed in moving/duplicating the fingers,
and rearranging the pointers incident to them.
Our notion of fingers is in principle the same as defined earlier for BSTs,
see, \eg, the lazy fingers~\cite{DemaineILO13}, but we prefer to explicitly 
distinguish between \emph{transient} and \emph{persistent} fingers:
\begin{definition}
	A data structure has access to $k$ \textit{persistent fingers} if it is
	able to track, as global variables, $k$ special pointers
	that it is able to retain across accesses.
	The algorithm is allowed to manipulate these fingers arbitrarily
	during the accesses.
	In contrast, we use \textit{transient finger} to denote fingers that
	only exist during a single operation, and are forgotten / reset
	before the next operation.
\end{definition}
Our definition of transient fingers is motivated by the
observation that a large number of persistent pointers 
trivializes most data structure questions~--
formalized in Appendix~\ref{app:elementary-offline-array-algorithm},
specifically Lemma~\ref{lem:very-very-lazy}~--
but the same is not true for transient fingers.
We can use transient fingers for tree rearrangement, 
but not for shortening the access path.
In particular, in the standard BST model, transient fingers
do not add any power to the algorithm, as local rearrangements
(rotations) are equally powerful there.

\subsection{Dynamic Optimality in Binary Search Trees}

The binary-search-tree model is another restricted pointer-machine model,
in which each node of a binary tree stores a key, and the keys have to fulfill
the search-tree property.
An execution in the model can move a finger around the tree or rotate an edge of the tree.
A variant of the BST model instead asks for specifying a \emph{replacement tree} 
for the subset of nodes visited while serving a request. 
Here, costs are measured by the number of visited nodes.
For BSTs, both models are equivalent (up to constant factors)~\cite{Kozma2016},
and so is the addition of further transient fingers.
For a general overview on dynamic optimality,
we refer the reader to Iacono's 2013 survey~\cite{Iacono2013} and 
the comprehensive introduction in Kozma's dissertation~\cite{Kozma2016}.
They discuss (instance-specific) upper bounds (\eg, the working-set bound), 
(instance-specific) lower bounds, 
and the state of knowledge on concrete algorithms, in particular
Splay~\cite{SleatorTarjan1985} and Greedy~\cite{Lucas1988,Munro2000},
as well as the geometric view of BST algorithms based
on satisfied point sets~\cite{DemaineHarmonIaconoKanePatrascu2009}.

It is easy to see that when we do \emph{not} know the accesses in advance, 
most access sequences will require costs in $\Omega(n\log n)$:
at any point in time all but $2^{\lfloor \lg n\rfloor/2} \le \sqrt n$ nodes 
are at depths $\ge \lfloor \lg n\rfloor / 2$ (and hence incur logarithmic cost to access).
In the offline model, when we do know all accesses in advance, this is not at all obvious, 
but it has been shown that there are ``universally hard'' access 
sequences that require $\Omega(n\log n)$ access cost
in \textit{any} binary-search-tree algorithm,
online or not~\cite{Wilber1989,DemaineHarmonIaconoKanePatrascu2009}.

An intriguing feature of the binary-search-tree model is that not only the
question about constant-competitive online algorithms remains wide open,
also the existence of (reasonably efficient) instance-optimal \emph{offline} algorithms is unsolved.
Indeed, designing good offline algorithms seems no simpler, and this fact
is seen as one reason of why a proof of (or counterexample for) dynamic optimality 
for splay trees has remained elusive~\cite{
	ChalermsookGKMS15,LevyTarjan2019}.
To add insult to the injury, a breakthrough result in the field
was that Greedy, the most promising candidate of an instance-optimal offline algorithm,
can indeed be turned into an \emph{online} algorithm 
(paying only a constant-factor increase in access costs, but a hefty fee in terms of 
conceptual complexity of the algorithm)~\cite{DemaineHarmonIaconoKanePatrascu2009}.
The current best upper bounds are $O(\log\log{n})$ 
competitive search trees~\cite{DemaineHarmonIaconoPatrascu2007,WangDS06}.

\paragraph{Standard vs Leaf-Oriented Trees}

The standard BST model stores a key in every node of a binary tree.
An alternative are leaf-oriented BSTs, where only leaves carry a key,
and internal nodes contain a copy of key serving only as ``routers'' for guiding searches.
Even though they are much less prominent than their cousins with keys in all nodes,
leaf-oriented BSTs have been studied, \eg, in the context
of concurrent data structures~\cite{EllenFatourouRuppertBreugel2010}, 
where their conceptual simplicity and the locality of pointer changes is valuable.

From the perspective of adaptive BSTs,
standard BSTs and leaf-oriented BSTs turn out to be equivalent:
there are constant-factor-overhead simulations for both directions.
The details are given in \wref{app:leaf-oriented-BSTs}.

\subsection{Unordered Binary Trees}

The work closest to ours are the articles by Fredman~\cite{Fredman2011,Fredman2012} 
mentioned above; his motivation, too, was to study self-adjusting heaps.
Clearly, the search-tree property is a useless restriction for priority-queue
implementations; but so seems insisting on \emph{binary} trees.
Indeed, both pairing heaps (an analog of splay trees in the priority-queue world)
and Fibonacci heaps 
are heap-ordered trees with arbitrary node degrees.
However, Fredman earlier showed that such forest-based heaps can~--
in some generality~-- be encoded as binary tournament trees:
In~\cite{Fredman1999}, he discusses how tournament trees 
can be simulated by forest-based heaps, 
and this mapping can also be used in reverse.

In tournament trees, all accesses are to leaves,
so if one was to consider the question whether pairing heaps 
or other self-adjusting heap variants~-- recast as tournament-tree rearrangement
heuristics~-- are (constant-)competitive algorithms, 
one should only demand competitiveness against accesses to leaves.
In Fredman's model of unordered binary trees, all nodes carry a key,
but we extend his arguments to the leaf-oriented tournament trees in this paper.

\subsection{Stable heaps}
\label{sec:stable-heaps}

In a recent work,
Kozma and Saranurak~\cite{KozmaSaranurak2018} set out
to establish a theory of instance-optimality for forest-oriented heaps 
(and in particular pairing heap variants).
They restrict access sequences on heaps to ``sorting mode''~--
$n$ inserts followed by $n$ extract-mins~-- and modify the primitive ``link'' operation
to be \emph{``stable'',} \ie, to always keep the left-to-right order of subtrees intact.

More specifically, after an initial sequence of $n$ inserts, the heap consists of
a list of $n$ top-level singleton roots.
Each of the following $n$ extract-min operations is served by
stably linking adjacent pairs of top-level roots, reducing their number by on each time,
until a single root is left (which contains the minimum).
The minimum is then removed, and its children form the new list of top-level roots.

The main result of Kozma and Saranurak is that every heap algorithms in 
this ``stable-heap'' model translates to a binary-search-tree algorithm, but critically, 
does not show that binary-search-tree algorithms imply heap ones.
Therefore, the connections exhibited in~\cite{KozmaSaranurak2018}
do not rule out the possibility that dynamic optimality holds
for binary search trees, but not for (stable) heaps.

An important observation about stable heaps is that the stability condition for links
implies that there are at most $C_n \le 4^n$ stable heap structures for a fixed insertion 
order of $n$ elements.
It is unclear what consequences this restriction of the freedom in rearrangements has
for algorithms.

\subsection{Further Related Work}

There are few other works that modify the computational model
to gain insight into the nature of the dynamic-optimality conjecture.
Iacono~\cite{Iacono2005} introduced the notion of ``key-independent optimality'',
in which costs are averaged over all possible orders of keys (key ranks chosen randomly).
He shows that any algorithm satisfying the working-set bound is optimal in the key-independent sense,
and hence so are Splay and Greedy.
The setup is different from our unordered trees 
since the maintained tree does have to conform to the search-tree
property once the order of keys has been chosen.
Bose et al.~\cite{Bose2008} studied dynamic optimality on skip lists and variants of B-trees.
They show that when insisting on certain balancing criteria, the working set bound
is actually a lower bound for serving an access sequence with these data structures.

To our knowledge, the systematic separation of online
and offline algorithms, or lower bounds for competitive ratios,
is relatively understudied.
Lower bounds for competitive data structures that we are aware
of only include deterministic paging algorithms~\cite{ManasseMS90},
and linear searches on lists under arbitrary
rearrangements of visited portions~\cite{MartinezRoura2000,Munro2000}.

Our modeling of heaps with the goal of providing an offline/online
separation is motivated by analogous results
on lists~\cite{MartinezRoura2000,Munro2000}, which gave a rearrangement
model where online algorithms must take $\Omega(n^2)$,
while offline algorithms take $O(n \log{n})$.
However, we spend significantly more, if not most, of our effort
addressing limitations on how the visited portion at each access
can be rearranged.
This is because of the much lower worst-case runtime upper
bound in the static case ($O(n \log{n})$ as opposed to $O(n^2)$
for move-to-front on lists): the $O(\log{n})$ overhead associated
with an arbitrary shuffle, or the $O(n \log{n})$ upper-bound
obtained from implementing a merge-sort like scheme~\cite{Munro2000}
is too high for heaps.

Our treatment of fingers follows the study of multi-finger
binary search trees by Demaine et al.~\cite{DemaineILO13}
and Chalermsook et al.~\cite{ChalermsookGKMS15}.
To our knowledge, aside from the restriction to rotations
made by Fredman~\cite{Fredman2011,Fredman2012}, which implicitly
assumes a constant number of fingers, the role of fingers
in heaps have not been explicitly studied previously.

\section{Online and Offline Separations}
\label{sec:models}

In this section, we prove our first result, \wref{thm:Main}.
As a warmup, we present a (simplified) counting argument for Fredman's ``Wilber-style'' 
lower bound.
We then provide two bounds for the competitive ratio of any online tournament-tree
algorithm, which are interesting for a small resp.\ large
numbers of fingers.

\subsection{Information-Theoretic Wilber Bound}

Fredman proved for his model of unordered binary trees,
that some access sequences require cost $\Omega(n \log n)$ 
to serve (Theorem~3 in~\cite{Fredman2012}).
We extend his result to tournament trees with $k$ fingers.

\begin{theorem}[Wilber-style lower bound]
\label{thm:wilber0}
	For any $n$ and $m$ there is an access sequence $A \in [n]^m$
	that requires total cost at least
	\[
		m \cdot \log_{10k}(n)
		\wwrel=
		 m \cdot \log_k(n) \left(1\pm\Oh\left(\frac1{\log k}\right)\right)
	\]
	in any tournament tree with $k$ fingers, even offline and with persistent fingers.
\end{theorem}

\begin{proof}
The proof is a counting argument.
We can encode any sequence of $t$ operations 
(of the allowed operations as defined in \wref{def:model}) 
in a tournament tree with $k$ fingers
by specifying for each time step, 
which of the $k$ fingers we used and 
which of the 10 possible operations we executed.
Given the sequence of operations and the initial tree, 
we can uniquely reconstruct the access sequence $A$
that was served by it (by virtue of the ``serve'' operations).

In total, there are $(10k)^t$ sequences of operations with cost $t$, 
from which we can reconstruct at most $(10k)^t$ different access sequences
that can be served with cost $t$;
(some encodings represent an invalid execution and do not correspond to a
served access sequence).
Note that we can always add dummy operations to 
an operations sequence with cost $<t$ to turn it into one of 
length exactly $t$ that serves the same access sequence $A$, 
so it suffices to count the latter ones.

Since there are $n^m$ different access sequences of length $m$ on $n$ keys, 
we can only serve all correctly when $(10k)^t \ge n^m$, or when
$t \ge m \log_{10k}(n)$.
\end{proof}

This means, the best amortized cost per access to hope for is $\Theta(\log_k(n))$ 
(in the worst case).
Our offline algorithm in \wref{sec:loglog} will essentially achieve that.
The above proof also shows that half of all access sequences require cost
$\ge m \log_{10k}(n/2)$ etc., so $\log_k(n)$ is indeed a lower bound for
the \emph{average} amortized cost, as well.

\subsection{Few Fingers}

We extend the rotation-based
argument by Fredman~\cite{Fredman2011} to account for all possible
operations involving the fingers as defined in \wref{def:model}.
We first generalize the key lemma from Fredman's lower bound~\cite[Lem.\,2]{Fredman2011}.

\begin{lemma}[Adversarial Permutations]
\label{lem:magic-permutation}
	For any $n$, any sufficiently large $b \leq n$ and $k = o(b)$ (as $b\to\infty$),
	we can find a (fixed) ``adversarial'' permutation
	$\pi$ on  $[1 \ldots b]$ such that for any initial configuration $(T,I,B)$ of a
	tree $T$ on $[1 \ldots n]$,
	locations $I$ of $k$ (persistent) fingers,
	and access sequence $B = a_1, \ldots, a_{b}$ of $b$ distinct accesses
	in $T$,
	wither the sequence $B$ itself
	or the permuted sequence $B_\pi = a_{\pi(1)}, \ldots, a_{\pi(b)}$ requires cost
	at least $0.3 b \log_{10k}(b)$, 
	even offline and using $k$ persistent fingers.
\end{lemma}

The proof is an extension of Fredman's argument.
We first segment out the fact that candidates for adversarial permutations 
can be refuted with small trees.

\begin{lemma}[Small counterexample trees]
	\label{lem:small-trees}
	For any values $b$ and any value $t$,
	if $\pi$ is a permutation such that
	there exists a tree $T$ on $[n]$ with $n \geq b$
	and initial positions $I \in [n]^k$ for $k$ (persistent) fingers in $T$,
	as well as an access sequence $B= a_1, \ldots, a_b$
	such that both $a_1, \ldots, a_b$ and
	$B_\pi = a_{\pi(1)}, \ldots, a_{\pi(b)}$ 
	can be served starting from $(T,I)$ with cost at most $t$,
	then there exists a tree $T'$ over $\mathcal N \subseteq [n]$ 
	containing $1,\ldots,b$
	on at most $t' = |\mathcal N| \le 2t$ vertices
	so that both $a_1, \ldots, a_b$
	and $a_{\pi(1)}, \ldots, a_{\pi(b)}$ can be served starting from $(T',I)$
	with total cost at most $t$.
\end{lemma}

\begin{proof}
The proof is analogous to the second part of Fredman's
proof, but with the role of root replaced by the fingers:
In short, an operation sequences of cost $t$ 
can touch at most $t$ vertices on top of the accessed nodes
$a_1,\ldots,a_b$, so we cannot see more than a limited neighborhood 
of these nodes.

More specifically, both the 
operations sequence $S_1$, serving $B$,
and $S_2$, serving $B_\pi$,
can each visit a portion of the tree of size at most $t$, 
and that portion must contain the initial
positions of all fingers and $a_1,\ldots,a_b$.
With persistent fingers, we allow a search to start at any finger
(instead of the root), so the visited region is potentially disconnected,
but for each execution sequence, it consists of the union of $k$ subtrees,
since the region explored by one finger 
(before potentially jumping to the location of another finger)
is a connected region.

We now consider (induced subtree of) the union $\mathcal N$ 
of the nodes in these $2k$ regions.
If the result is not a connected graph, 
we arbitrarily connect the components
(attaching one component as the child of any leaf of another),
forming a single connected binary tree $T'$ over $t' \le 2t$ nodes.
$S_1$ and $S_2$ are still valid executions when starting with $(T',I)$ 
instead of $(T,I)$, proving the claim.~
\end{proof}

We now perform a counting argument similar to the first
part of Fredman's proof~\cite{Fredman2011},
but taking into account the locations of fingers as well.

\begin{proof}[\wref{lem:magic-permutation}]
Our goal will be to enumerate and count all possible witnesses $(T,I,B)$
to the ``tameness'' of a some permutation $\pi$ over $[b]$, \ie,
initial configurations such that when starting with tree $T$ and fingers at $I$,
both $B$ and $B_\pi$ require at most $t$ operations to serve.
Since each witness can eliminate at most one candidate for the adversarial permutation,
having fewer than $b!$ witnesses implies the claimed existence of $\pi$;
we will show that for $t$ bounded as in the lemma, this is indeed true.

By definition, for any witness $(T,I,B)$ to the tameness of $\pi$,
there are two operation sequences $S_1$ and $S_2$, both of length at most $t$,
so that $S_1$ serves $B$ and $S_2$ serves $B_\pi$.
Moreover, we can recover both $B$ and $B_\pi$, and hence $\pi$ itself, from
$(T,I,S_1,S_2)$.
We therefore obtain a (crude) over-approximation of the set of witnesses
by counting all such quadruples.
Now, by \wref{lem:small-trees} we can restrict our attention to trees $T$
over $2t$ nodes, and there are no more than $4^{2t}$ such;
there are $(2t)^k$ choices for the initial positions of $k$ pointers,
and $(10k)^t$ options for $S_1$ and $S_2$, for a total of 
\begin{align*}
		W(t)
	&\wwrel=
		2^{4t} \cdot (2t)^k \cdot (10k)^{2t}
\end{align*}
witness candidates.
(The actual number of witnesses is lower because not all of these quadruples encode a valid witness.)
We will now show that for $t = (1-\epsilon) c b\log_{10k}(b)$ with $\epsilon>0$ fixed and 
$c = (2+4/\lg10)^{-1} \approx 0.3121$, we have that
$\lg W(t) < (1-\epsilon) b \lg b$ for large enough $b$, and hence eventually also $W(t) < b!$:
\begin{align*}
		\lg W(t)
	&\wwrel=
		4t + k (\lg(t)+1) + 2t \lg (10k)
\\	&\wwrel{\relwithtext[r]{$[k = o(b)]$}=}
		2t \lg (10k) + 4t + o(b \log(t))
\\	&\wwrel{\relwithtext[r]{[insert t]}\le}
		(2 + \tfrac 4{\lg 10} )c \cdot b\lg(b) + o(b \log(b))
\\	&\wwrel=
		(1-\epsilon) b\lg(b) + o(b \log(b)).
\end{align*}
Thus for $t$ bounded as in the statement and sufficiently large $b$, we have $\lg W(t) < b!$,
so there are fewer than witnesses to tameness than there are permutations, 
and hence an adversarial permutation must exists.
\end{proof}

Fredman~\cite{Fredman2011} used such a permutation to adaptively construct
an adversarial sequence for any online algorithm.
We can readily check that the construction also applies to tournament heaps
with $k$ fingers, and therefore obtain a lower bound for the few-finger case.

\begin{corollary}[Few Fingers]
	\label{cor:small-k}
	The competitiveness ratio of tournament heaps with
	$k$ fingers is at least $\Omega(\log_k{n}) = \Omega(\frac{\log{n}}{\log{k}})$.
\end{corollary}
This result even holds when the online algorithm has access to $k$ persistent fingers,
whereas the offline algorithm is restricted to a single transient finger.

\begin{proof}
	We follow Fredman's proof~\cite{Fredman2011}.
	Specifically, we assume $n = b^2$ is a square,
	and let $\pi$ be an ``adversarial'' permutation as given 
	by Lemma~\ref{lem:magic-permutation} with length $b$.
	\Withoutlossofgenerality let the keys in the initial tree 
	by $1,\ldots,n$ in inorder.
	Then for any online algorithm $\textsc{Alg}$, we describe an adversary
	that (adaptively) generates a permutation $A$ on which $\textsc{Alg}$ incurs
	total cost $\Omega(n \log_{k}(n))$. For that, the adversary iteratively considers
	the $b$ elements in $B_i = [(i - 1)b + 1 \dots ib] = (i-1)b + [1\ldots b]$ for $i=1,\ldots,b$, 
	and selects as the next block of requests 
	either $B_i$ itself (\ie, requests in sorted order) or 
	its permuted copy 
	\[
			B_i^\pi 
		\wwrel= 
		\left(i - 1\right)b + \pi\left( 1\right),\;
		\left(i - 1\right)b + \pi\left( 2\right),\;
		\ldots,\;
		\left(i - 1\right)b + \pi\left( b \right).
	\]
	Call the $i$th block of requests $C_i$.
	By \wref{lem:magic-permutation}, there is always a choice for the adversary that makes
	$\textsc{Alg}$ pay $\Omega(b \log_k (b)) = \Omega(\sqrt n \log_k (n))$ on $C_i$,
	for a total cost of $\Omega(n\log_k (n))$ after all $\sqrt n$ blocks in $A$.
	
	It remains to check that the resulting permutation $A$
	can be processed in $O(n)$ operations when known offline, even with a single
	transient finger.
	This is to be contrasted with the (potentially) superconstant number $k$ of 
	(persistent) fingers that the online algorithm was allowed to use.
	
	We again follow Fredman's proof, but greatly simplify the presentation 
	based on the more recent understandings of multi-finger
	search trees~\cite{DemaineILO13,ChalermsookGKMS15}:
	We describe an offline algorithm for \emph{two} (transient) fingers
	instead of a single finger;
	since we can simulate a fixed, constant number of fingers
	with a single one with a constant-factor overhead~\cite{DemaineILO13,ChalermsookGKMS15}, 
	this yields the desired result.
	
	Our strategy to serve $A$ will be to spend $O(n)$ overhead upon
	the first access and thereby transform $T$ into a path with keys sorted
	by next access. All future accesses can then be served by
	simply rotating one edge at the root each.
	Recall that $A$ is the concatenation of $C_1,C_2,\ldots,C_b$,
	where each $C_i$ is either $B_i$ or $B_i^\pi$.
	So we start with the tree that is $1, \ldots, n = b^2$ on a path;
	we can rearrange any initial tree with $O(n)$ rotations into such a path.
	It suffices to apply $\pi$ to the permuted blocks (in the tree)
	to obtain a tree from which $A$ can be served in $O(n)$ steps.
	For this, we can use two fingers to arrange the trees into
	an analog of operations on 2-D arrays:
	\begin{enumerate}
		\item We first transform the single chain into a ``row major'' matrix,
		that is, each set $[(i - 1)b + 1, ib]$ forms a path,
		and the root has a path containing the roots $1, b + 1, 2b + 1, \ldots$ of these block paths.
		\wref{fig:step1} illustrates this step.
		Using two fingers, one for ``reading'' through the path and the other for appending to the current
		row, this transformation is easily accomplished with $O(n)$ operations.
		
		\begin{figure}[htbp]
			\scalebox{.65}{\begin{tikzpicture}[xscale = 1,yscale=1.3]
			\medmuskip=1mu\smaller[0]
			
			\begin{scope}[shift={(-8,1)}]
				\foreach \i/\s/\l in {1/treenode/1,2/treenode/2,3/treenode/3,4/treenode/4,5/empty/$\dots$,6/treenode/$n-1$,7/treenode/$n$} {
					\node[\s] (p\i) at (-1.5*\i,-\i) {\l} ;
				}
				\foreach \i/\ii in {1/2,2/3,3/4,4/5,5/6,6/7} {
					\draw (p\i) -- (p\ii) ;
				}
			\end{scope}
			
			\node[scale=3] at (-9,-3) {$\leadsto$};
			
			\node[treenode] (x1) at (0,0) {$1$} ;
			\node[treenode] (x2) at (-1.5, -1) {$b+1$} ;
			\node[treenode] (x3) at (-3, -2) {$2b+1$} ;
			\node[treenode] (x4) at (-4.5, -3) {$3b+1$} ;
			\node[treenode] (x5) at (-6, -4) {$4b+1$} ;
			\node[empty] (xetc) at (-7.5, -5) {$\ldots$} ;
			
			\draw (x1) to (x2);
			\draw (x2) to (x3);
			\draw (x3) to (x4);
			\draw (x4) to (x5);
			\draw (x5) to (xetc);
			
			\node[treenode] (x11) at (1.5,-1) {$2$} ;
			\node[empty] (x1etc) at (3,-2) {$\ldots$} ;
			\node[treenode] (x1b) at (4.5,-3) {$b$} ;
			\draw (x1) to (x11);
			\draw (x11) to (x1etc);
			\draw (x1etc) to (x1b);
			
			\node[treenode] (x21) at (0,-2) {$b+2$} ;
			\node[empty] (x2etc) at (1.5,-3) {$\ldots$} ;
			\node[treenode] (x2b) at (3,-4) {$2b$} ;
			\draw (x2) to (x21);
			\draw (x21) to (x2etc);
			\draw (x2etc) to (x2b);
	
			\node[treenode] (x31) at (-1.5,-3) {$2b+2$} ;
			\node[empty] (x3etc) at (0,-4) {$\ldots$} ;
			\node[treenode] (x3b) at (1.5,-5) {$3b$} ;
			\draw (x3) to (x31);
			\draw (x31) to (x3etc);
			\draw (x3etc) to (x3b);
	
			\node[treenode] (x41) at (-3,-4) {$3b+2$} ;
			\node[empty] (x4etc) at (-1.5,-5) {$\ldots$} ;
			\node[treenode] (x4b) at (0,-6) {$4b$} ;
			\draw (x4) to (x41);
			\draw (x41) to (x4etc);
			\draw (x4etc) to (x4b);
	
			\node[treenode] (x51) at (-4.5,-5) {$4b+2$} ;
			\node[empty] (x5etc) at (-3,-6) {$\ldots$} ;
			\node[treenode] (x5b) at (-1.5,-7) {$5b$} ;
			\draw (x5) to (x51);
			\draw (x51) to (x5etc);
			\draw (x5etc) to (x5b);
			\end{tikzpicture}}
			\caption{%
				The first step of the transformation, from path to row-major ordered matrix.
				The shaded nodes each consist of one or two internal nodes and a leaf with the stored key.%
			}
			\label{fig:step1}
		\end{figure}
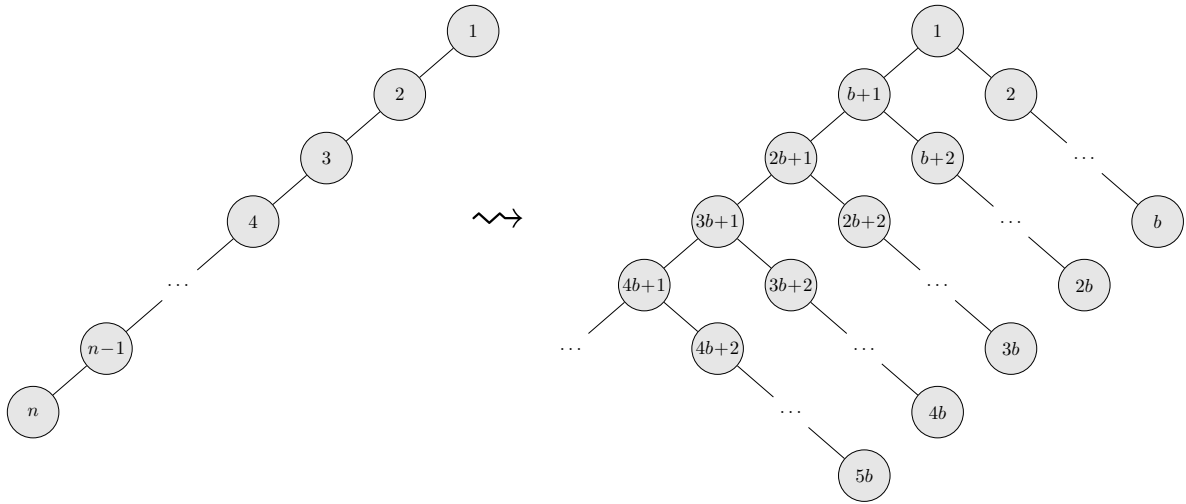
		
		\item Next, extract out all the $i$s for which
		we have to apply $\pi$ to obtain $C_i$.
		We refer this subset as $\widehat{I}$. In the tree, we rearrange
		the path containing the heads of the $B_i$s so that all
		$\widehat{I}$-blocks appears as a prefix; see \wref{fig:step2}
		This task is very similar to a quicksort-style partition on linked lists.
		Using one ``read finger'' and one ``write finger'' to traverse the list of heads
		in parallel, it is achieved with $O(b)$ operations.
		
		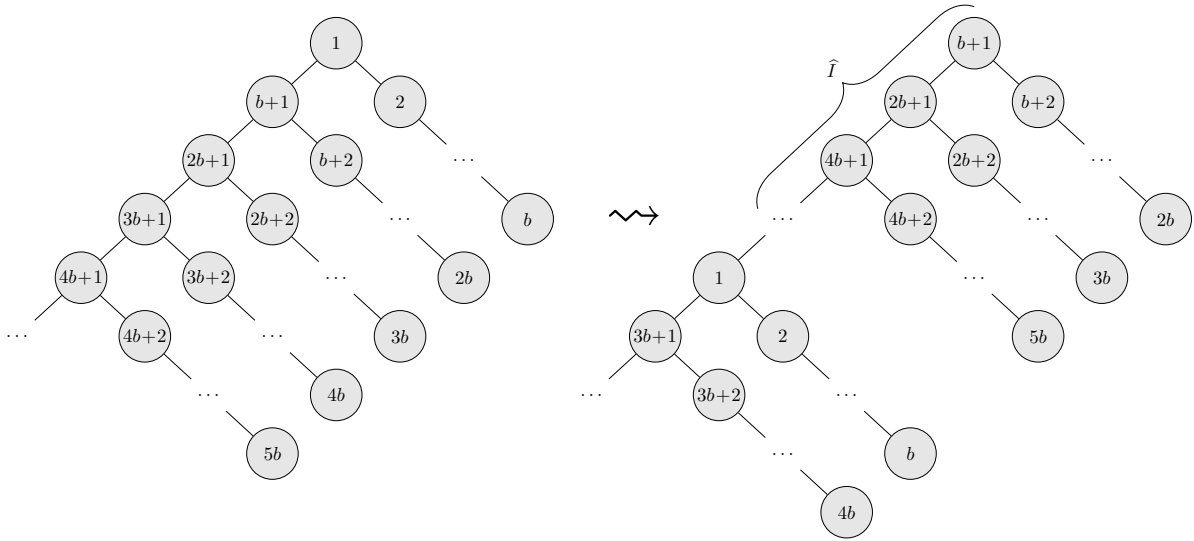
\begin{figure}[htb]
			\scalebox{.65}{\begin{tikzpicture}[xscale = 0.87,yscale=1.2]
				\medmuskip=1mu\smaller[0]
				\node[scale=3] at (-8,-3) {$\leadsto$};
				\begin{scope}[shift={(-15,0)}]
					\node[treenode] (x1) at (0,0) {$1$} ;
					\node[treenode] (x2) at (-1.5, -1) {$b+1$} ;
					\node[treenode] (x3) at (-3, -2) {$2b+1$} ;
					\node[treenode] (x4) at (-4.5, -3) {$3b+1$} ;
					\node[treenode] (x5) at (-6, -4) {$4b+1$} ;
					\node[empty] (xetc) at (-7.5, -5) {$\ldots$} ;
					
					\draw (x1) to (x2);
					\draw (x2) to (x3);
					\draw (x3) to (x4);
					\draw (x4) to (x5);
					\draw (x5) to (xetc);
					
					\node[treenode] (x11) at (1.5,-1) {$2$} ;
					\node[empty] (x1etc) at (3,-2) {$\ldots$} ;
					\node[treenode] (x1b) at (4.5,-3) {$b$} ;
					\draw (x1) to (x11);
					\draw (x11) to (x1etc);
					\draw (x1etc) to (x1b);
					
					\node[treenode] (x21) at (0,-2) {$b+2$} ;
					\node[empty] (x2etc) at (1.5,-3) {$\ldots$} ;
					\node[treenode] (x2b) at (3,-4) {$2b$} ;
					\draw (x2) to (x21);
					\draw (x21) to (x2etc);
					\draw (x2etc) to (x2b);
			
					\node[treenode] (x31) at (-1.5,-3) {$2b+2$} ;
					\node[empty] (x3etc) at (0,-4) {$\ldots$} ;
					\node[treenode] (x3b) at (1.5,-5) {$3b$} ;
					\draw (x3) to (x31);
					\draw (x31) to (x3etc);
					\draw (x3etc) to (x3b);
			
					\node[treenode] (x41) at (-3,-4) {$3b+2$} ;
					\node[empty] (x4etc) at (-1.5,-5) {$\ldots$} ;
					\node[treenode] (x4b) at (0,-6) {$4b$} ;
					\draw (x4) to (x41);
					\draw (x41) to (x4etc);
					\draw (x4etc) to (x4b);
			
					\node[treenode] (x51) at (-4.5,-5) {$4b+2$} ;
					\node[empty] (x5etc) at (-3,-6) {$\ldots$} ;
					\node[treenode] (x5b) at (-1.5,-7) {$5b$} ;
					\draw (x5) to (x51);
					\draw (x51) to (x5etc);
					\draw (x5etc) to (x5b);
				\end{scope}
				
				\node[treenode] (x2) at (0,0) {$b+1$} ;
				\node[treenode] (x3) at (-1.5, -1) {$2b+1$} ;
				\node[treenode] (x5) at (-3, -2) {$4b+1$} ;
				\node[empty] (xetc1) at (-4.5, -3) {$\ldots$} ;
				\node[treenode] (x1) at (-6, -4) {$1$} ;
				\node[treenode] (x4) at (-7.5, -5) {$3b+1$} ;
				\node[empty] (xetc2) at (-9, -6) {$\ldots$} ;
				
				\draw (x2) to (x3);
				\draw (x3) to (x5);
				\draw (x5) to (xetc1);
				\draw (xetc1) to (x1);
				\draw (x1) to (x4);
				\draw (x4) to (xetc2);
				
				\draw[decorate,decoration={brace,mirror,amplitude=20pt}] 
					($(x2.north) + (0,5pt)$) -- ($(xetc1.west)+(-5pt,5pt)$) node [black,midway,xshift=-20pt,yshift=23pt]{$\widehat{I}$};
				
				\node[treenode] (x21) at (1.5,-1) {$b+2$};
				\node[empty] (x2etc) at (3,-2) {$\ldots$};
				\node[treenode] (x2b) at (4.5,-3) {$2b$};
		
				\node[treenode] (x31) at (0,-2) {$2b+2$};
				\node[empty] (x3etc) at (1.5,-3) {$\ldots$};
				\node[treenode] (x3b) at (3,-4) {$3b$};
		
				\node[treenode] (x51) at (-1.5,-3) {$4b+2$};
				\node[empty] (x5etc) at (0,-4) {$\ldots$};
				\node[treenode] (x5b) at (1.5,-5) {$5b$};

				\node[treenode] (x11) at (-4.5,-5) {$2$};
				\node[empty] (x1etc) at (-3,-6) {$\ldots$};
				\node[treenode] (x1b) at (-1.5,-7) {$b$};
				
				\node[treenode] (x41) at (-6,-6) {$3b+2$};
				\node[empty] (x4etc) at (-4.5,-7) {$\ldots$};
				\node[treenode] (x4b) at (-3,-8) {$4b$};
				
				\draw (x1) to (x11);
				\draw (x11) to (x1etc);
				\draw (x1etc) to (x1b);
				\draw (x2) to (x21);
				\draw (x21) to (x2etc);
				\draw (x2etc) to (x2b);
				\draw (x3) to (x31);
				\draw (x31) to (x3etc);
				\draw (x3etc) to (x3b);
				\draw (x4) to (x41);
				\draw (x41) to (x4etc);
				\draw (x4etc) to (x4b);
				\draw (x5) to (x51);
				\draw (x51) to (x5etc);
				\draw (x5etc) to (x5b);			
				
			\end{tikzpicture}}
			\caption{%
				Step 2 of the transformation: separating permuted and non-permuted blocks.
				In the example, the indices of permuted blocks is $\widehat{I} = \{2, 3, 5, \ldots \}$.
			}
			\label{fig:step2}
		\end{figure}
		
		\item Now we ``transpose''
		this prefix into a ``column-major'' ordering:
		there is a path starting from the root containing all values
		of $1 \leq j \leq b$, and all elements of the form
		$[ib + j]$ for $i \in \widehat{I}$ are attached to $j$
		in a path; see \wref{fig:step3} for an example.
		
		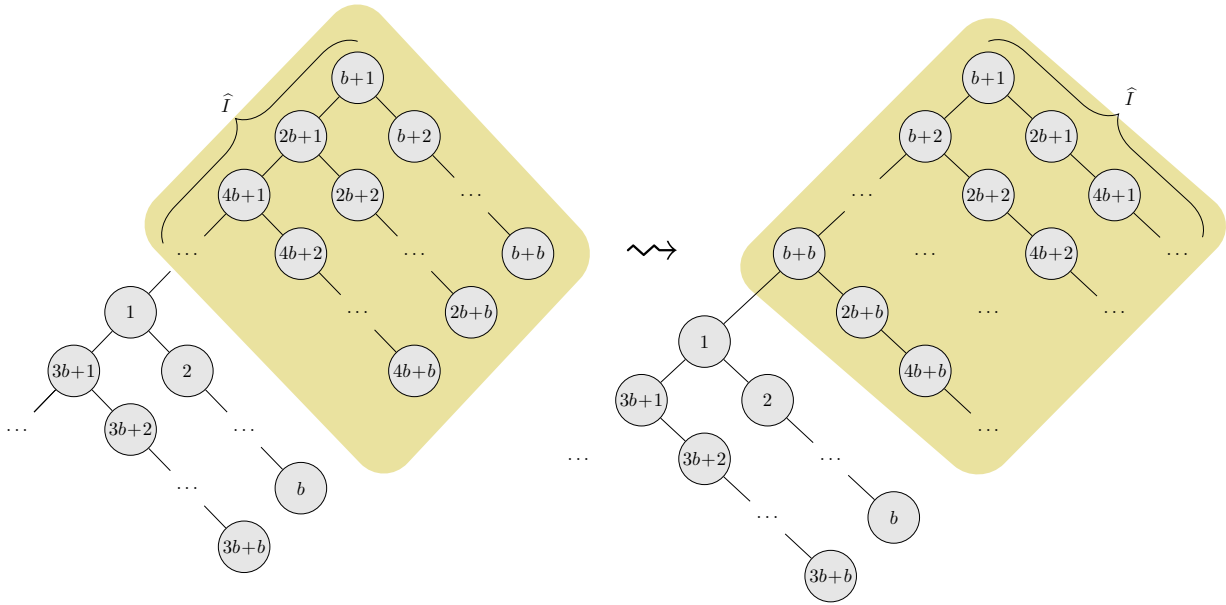
\begin{figure}[htb]
			\scalebox{.65}{\begin{tikzpicture}[xscale = 0.86,yscale=1.2]
				\medmuskip=1mu\smaller[0]
				\node[scale=3] at (-8,-3) {$\leadsto$};
				\colorlet{mm}{yellow!80!black!50}
				
				\begin{scope}[shift={(-15,0)},xscale=.9]
					
					\fill[mm,rounded corners=20pt] 
						(0,1.5) -- (-6,-2.5) -- (0.7,-7) -- (6.5,-3) --cycle
					;
					
					\node[treenode] (x2) at (0,0) {$b+1$} ;
					\node[treenode] (x3) at (-1.5, -1) {$2b+1$} ;
					\node[treenode] (x5) at (-3, -2) {$4b+1$} ;
					\node[empty] (xetc1) at (-4.5, -3) {$\ldots$} ;
					\node[treenode] (x1) at (-6, -4) {$1$} ;
					\node[treenode] (x4) at (-7.5, -5) {$3b+1$} ;
					\node[empty] (xetc2) at (-9, -6) {$\ldots$} ;
					
					\draw (x2) to (x3);
					\draw (x3) to (x5);
					\draw (x5) to (xetc1);
					\draw (xetc1) to (x1);
					\draw (x1) to (x4);
					\draw (x4) to (xetc2);
					
					\draw[decorate,decoration={brace,mirror,amplitude=20pt}] 
						($(x2.north) + (0,5pt)$) -- ($(xetc1.west)+(-5pt,5pt)$) node [black,midway,xshift=-20pt,yshift=23pt]{$\widehat{I}$};
					
					\node[treenode] (x21) at (1.5,-1) {$b+2$};
					\node[empty] (x2etc) at (3,-2) {$\ldots$};
					\node[treenode] (x2b) at (4.5,-3) {$b+b$};
			
					\node[treenode] (x31) at (0,-2) {$2b+2$};
					\node[empty] (x3etc) at (1.5,-3) {$\ldots$};
					\node[treenode] (x3b) at (3,-4) {$2b+b$};
			
					\node[treenode] (x51) at (-1.5,-3) {$4b+2$};
					\node[empty] (x5etc) at (0,-4) {$\ldots$};
					\node[treenode] (x5b) at (1.5,-5) {$4b+b$};

					\node[treenode] (x11) at (-4.5,-5) {$2$};
					\node[empty] (x1etc) at (-3,-6) {$\ldots$};
					\node[treenode] (x1b) at (-1.5,-7) {$b$};
					
					\node[treenode] (x41) at (-6,-6) {$3b+2$};
					\node[empty] (x4etc) at (-4.5,-7) {$\ldots$};
					\node[treenode] (x4b) at (-3,-8) {$3b+b$};
					
					\draw (x1) to (x11);
					\draw (x11) to (x1etc);
					\draw (x1etc) to (x1b);
					\draw (x2) to (x21);
					\draw (x21) to (x2etc);
					\draw (x2etc) to (x2b);
					\draw (x3) to (x31);
					\draw (x31) to (x3etc);
					\draw (x3etc) to (x3b);
					\draw (x4) to (x41);
					\draw (x41) to (x4etc);
					\draw (x4etc) to (x4b);
					\draw (x5) to (x51);
					\draw (x51) to (x5etc);
					\draw (x5etc) to (x5b);			
				\end{scope}

				\fill[mm,rounded corners=20pt] 
					(0,1.25) -- (-6.25,-3.25) -- (-.2,-7) -- (6,-2.5) --cycle
				;
	
				\node[treenode] (x2) at (0,0) {$b+1$} ;
				\node[treenode] (x3) at (1.5, -1) {$2b+1$} ;
				\node[treenode] (x5) at (3, -2) {$4b+1$} ;
				\node[empty] (xetc1) at (4.5, -3) {$\ldots$} ;
				
				\node[treenode] (x21) at (-1.5,-1) {$b+2$};
				\node[empty] (x2etc) at (-3,-2) {$\ldots$};
				\node[treenode] (x2b) at (-4.5,-3) {$b+b$};
		
				\draw (x2) to (x3);
				\draw (x3) to (x5);
				\draw (x5) to (xetc1);
				
				\draw (x4) to (xetc2);
				
				\draw[decorate,decoration={brace,amplitude=20pt}] 
					($(x2.north)+(0,5pt)$) -- ($(xetc1.east)+(5pt,8pt)$) 
					node [black,midway,xshift=20pt,yshift=25pt]{$\widehat{I}$};
						
				\node[treenode] (x31) at (0,-2) {$2b+2$};
				\node[empty] (x3etc) at (-1.5,-3) {$\ldots$};
				\node[treenode] (x3b) at (-3,-4) {$2b+b$};
				
				\node[treenode] (x51) at (1.5,-3) {$4b+2$};
				\node[empty] (x5etc) at (0,-4) {$\ldots$};
				\node[treenode] (x5b) at (-1.5,-5) {$4b+b$};
				
				\node[empty] (xetc1) at (3,-4) {$\ldots$};
				\node[empty] (xetcb) at (0,-6) {$\ldots$};

				\begin{scope}[shift={(-.75,-.5)}]
					\node[treenode] (x1) at (-6, -4) {$1$} ;
					\node[treenode] (x4) at (-7.5, -5) {$3b+1$} ;
					\node[empty] (xetc2) at (-9, -6) {$\ldots$} ;
					\draw (x2b) to (x1);
					\draw (x1) to (x4);
					
					\node[treenode] (x11) at (-4.5,-5) {$2$};
					\node[empty] (x1etc) at (-3,-6) {$\ldots$};
					\node[treenode] (x1b) at (-1.5,-7) {$b$};
					
					\node[treenode] (x41) at (-6,-6) {$3b+2$};
					\node[empty] (x4etc) at (-4.5,-7) {$\ldots$};
					\node[treenode] (x4b) at (-3,-8) {$3b+b$};
					
					\draw (x2) to (x21);
					\draw (x21) to (x2etc);
					\draw (x2etc) to (x2b);
					
					\draw (x21) to (x31);
					\draw (x31) to (x51);
					\draw (x51) to (xetc1);
				
					\draw (x2b) to (x3b);
					\draw (x3b) to (x5b);
					\draw (x5b) to (xetcb);
					
					\draw (x1) to (x11);
					\draw (x11) to (x1etc);
					\draw (x1etc) to (x1b);
					\draw (x4) to (x41);
					\draw (x41) to (x4etc);
					\draw (x4etc) to (x4b);
				\end{scope}
			\end{tikzpicture}}
			\caption{%
				Step 3 of the transformation: transposing the $\widehat I$ rows.
				The transposed part is highlighted; the other blocks remain unchanged.
			}
			\label{fig:step3}
		\end{figure}
		
		\item Now we apply $\pi$ to the $b$ first path heads, thereby applying $\pi$ in parallel
			to all blocks in $\widehat I$.
			The important observation is that it is the same permutation $\pi$ that has to be applied
			to all paths, so we can do it in one shot after the above preparation.
			An arbitrary permutation of $b$ nodes can be applied with 
			$O(b \log{b}) = O(n)$ operations (see, \eg, \wref{lem:larger-k-perm} below);
			here it would actually be sufficient to simulate, say, bubble sort using two fingers.
			\wref{fig:step4} shows the result.
		
		\begin{figure}[htb]
			\scalebox{.65}{\begin{tikzpicture}[xscale = 0.83,yscale=1.2]
				\medmuskip=1mu\smaller[0]
				\node[scale=3] at (-8,-3) {$\leadsto$};
				\colorlet{mm}{yellow!80!black!50}
				
				\begin{scope}[shift={(-15,0)}]
					\node[treenode] (x2) at (0,0) {$b+1$} ;
					\node[treenode] (x3) at (1.5, -1) {$2b+1$} ;
					\node[treenode] (x5) at (3, -2) {$4b+1$} ;
					\node[empty] (xetc1) at (4.5, -3) {$\ldots$} ;
					
					\node[treenode] (x21) at (-1.5,-1) {$b+2$};
					\node[empty] (x2etc) at (-3,-2) {$\ldots$};
					\node[treenode] (x2b) at (-4.5,-3) {$b+b$};
			
					\draw (x2) to (x3);
					\draw (x3) to (x5);
					\draw (x5) to (xetc1);
					\draw (x4) to (xetc2);
					
					\draw[decorate,decoration={brace,amplitude=20pt}] 
						($(x2.north)+(0,5pt)$) -- ($(xetc1.east)+(5pt,8pt)$) 
						node [black,midway,xshift=20pt,yshift=25pt]{$\widehat{I}$};
							
					\node[treenode] (x31) at (0,-2) {$2b+2$};
					\node[empty] (x3etc) at (-1.5,-3) {$\ldots$};
					\node[treenode] (x3b) at (-3,-4) {$2b+b$};
					
					\node[treenode] (x51) at (1.5,-3) {$4b+2$};
					\node[empty] (x5etc) at (0,-4) {$\ldots$};
					\node[treenode] (x5b) at (-1.5,-5) {$4b+b$};
					
					\node[empty] (xetc1) at (3,-4) {$\ldots$};
					\node[empty] (xetcb) at (0,-6) {$\ldots$};

					\begin{scope}[shift={(-.75,-.5)}]
						\node[treenode] (x1) at (-6, -4) {$1$} ;
						\node[treenode] (x4) at (-7.5, -5) {$3b+1$} ;
						\node[empty] (xetc2) at (-9, -6) {$\ldots$} ;
						\draw (x2b) to (x1);
						\draw (x1) to (x4);
						
						\node[treenode] (x11) at (-4.5,-5) {$2$};
						\node[empty] (x1etc) at (-3,-6) {$\ldots$};
						\node[treenode] (x1b) at (-1.5,-7) {$b$};
						
						\node[treenode] (x41) at (-6,-6) {$3b+2$};
						\node[empty] (x4etc) at (-4.5,-7) {$\ldots$};
						\node[treenode] (x4b) at (-3,-8) {$3b+b$};
						
						\draw (x2) to (x21);
						\draw (x21) to (x2etc);
						\draw (x2etc) to (x2b);
						
						\draw (x21) to (x31);
						\draw (x31) to (x51);
						\draw (x51) to (xetc1);
					
						\draw (x2b) to (x3b);
						\draw (x3b) to (x5b);
						\draw (x5b) to (xetcb);
						
						\draw (x1) to (x11);
						\draw (x11) to (x1etc);
						\draw (x1etc) to (x1b);
						\draw (x4) to (x41);
						\draw (x41) to (x4etc);
						\draw (x4etc) to (x4b);
					\end{scope}
				\end{scope}
				
			\medmuskip=0mu
			\tikzset{
				treenode/.append style={},
				streenode/.style={treenode,minimum size=37pt,font=\smaller},
			}
			
			\begin{scope}[shift={(0,0)}]
				\node[streenode] (x2) at (0,0) {$b+\pi(1)$} ;
				\node[streenode] (x3) at (1.5, -1) {$2b+\pi(1)$} ;
				\node[streenode] (x5) at (3, -2) {$4b+\pi(1)$} ;
				\node[empty] (xetc1) at (4.5, -3) {$\ldots$} ;
				
				\node[streenode] (x21) at (-1.5,-1) {$b+\pi(2)$};
				\node[empty] (x2etc) at (-3,-2) {$\ldots$};
				\node[streenode] (x2b) at (-4.5,-3) {$b+\pi(b)$};
		
				\draw (x2) to (x3);
				\draw (x3) to (x5);
				\draw (x5) to (xetc1);
				\draw (x4) to (xetc2);
				
				\draw[decorate,decoration={brace,amplitude=20pt}] 
					($(x2.north)+(0,5pt)$) -- ($(xetc1.east)+(5pt,8pt)$) 
					node [black,midway,xshift=20pt,yshift=25pt]{$\widehat{I}$};
						
				\node[streenode] (x31) at (0,-2) {$2b+\pi(2)$};
				\node[empty] (x3etc) at (-1.5,-3) {$\ldots$};
				\node[streenode] (x3b) at (-3,-4) {$2b+\pi(b)$};
				
				\node[streenode] (x51) at (1.5,-3) {$4b+\pi(2)$};
				\node[empty] (x5etc) at (0,-4) {$\ldots$};
				\node[streenode] (x5b) at (-1.5,-5) {$4b+\pi(b)$};
				
				\node[empty] (xetc1) at (3,-4) {$\ldots$};
				\node[empty] (xetcb) at (0,-6) {$\ldots$};

				\begin{scope}[shift={(-.75,-.5)}]
					\node[treenode] (x1) at (-6, -4) {$1$} ;
					\node[treenode] (x4) at (-7.5, -5) {$3b+1$} ;
					\node[empty] (xetc2) at (-9, -6) {$\ldots$} ;
					\draw (x2b) to (x1);
					\draw (x1) to (x4);
					
					\node[treenode] (x11) at (-4.5,-5) {$2$};
					\node[empty] (x1etc) at (-3,-6) {$\ldots$};
					\node[treenode] (x1b) at (-1.5,-7) {$b$};
					
					\node[treenode] (x41) at (-6,-6) {$3b+2$};
					\node[empty] (x4etc) at (-4.5,-7) {$\ldots$};
					\node[treenode] (x4b) at (-3,-8) {$3b+b$};
					
					\draw (x2) to (x21);
					\draw (x21) to (x2etc);
					\draw (x2etc) to (x2b);
					
					\draw (x21) to (x31);
					\draw (x31) to (x51);
					\draw (x51) to (xetc1);
				
					\draw (x2b) to (x3b);
					\draw (x3b) to (x5b);
					\draw (x5b) to (xetcb);
					
					\draw (x1) to (x11);
					\draw (x11) to (x1etc);
					\draw (x1etc) to (x1b);
					\draw (x4) to (x41);
					\draw (x41) to (x4etc);
					\draw (x4etc) to (x4b);
				\end{scope}
			\end{scope}
			\end{tikzpicture}}
			\caption{%
				Step 4 of the transformation: Applying $\pi$ in parallel to all $\widehat I$-blocks.
				This step only affects the first $b$ paths; the other blocks remain unchanged.
			}
			\label{fig:step4}
		\end{figure}
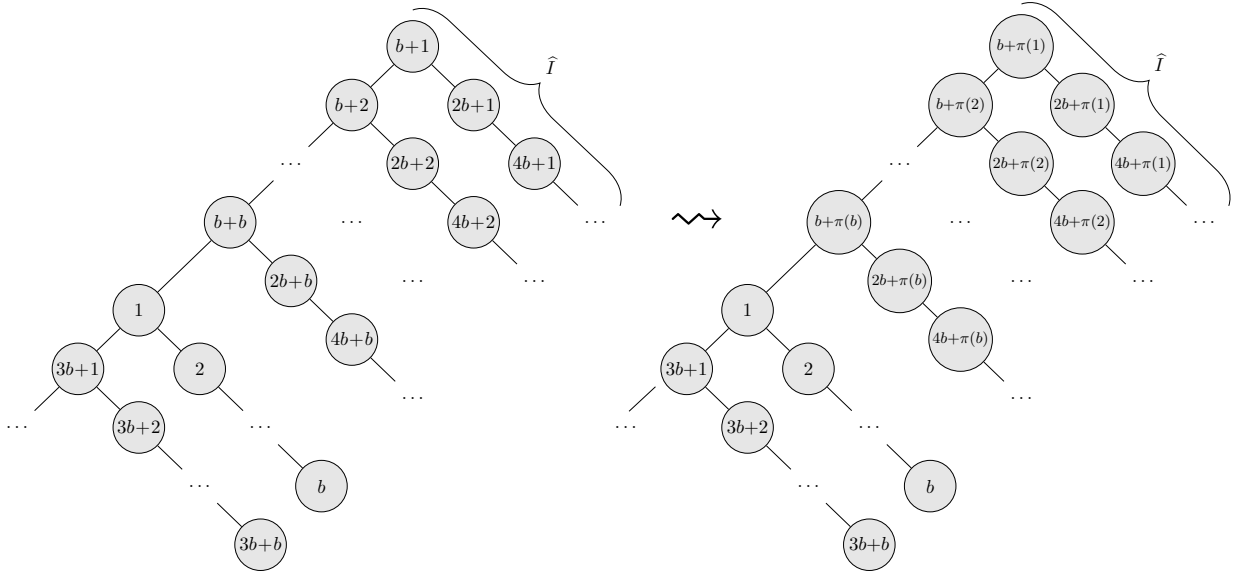
		
		\item Reversing the first two operations then ``stitches''
		this transformed sub-permutation back in the original,
		and has the overall effect of replacing $ib + j$
		by $ib + \pi(j)$ for all $i \in \widehat{I}$.
		The reversal of the transformations can be achieved at the same cost as argued above.
	\end{enumerate}
	
	In total, we have shown how to serve $A$ in $O(n)$ time offline, even with a single transient finger.
	Hence, the competitive ratio of $\textsc{Alg}$ is at least
	\[
		\Omega\left(	\frac{n \log_{k}n}{n} \right)
	\wwrel\geq
		\Omega\left( \log_{k}(n) \right).
	\]
\end{proof}

One can obviously repeat this process to obtain arbitrarily long access sequences with 
the same competitive ratio.

\subsection{Many Fingers}

A large number of fingers allow
us to efficiently implement arbitrary permutations, and hence, an 
\emph{order-by-next request} approach.
We consider here the simple case that the input is a permutation and 
show that we handle it far more efficiently than the
average cost of accessing a node in a tree with $n$ nodes.
This subroutine also forms the basis of our efficient offline
algorithm for arbitrary access sequences in \wref{sec:loglog}.

\begin{lemma}[Permute]
	\label{lem:larger-k-perm}
	Using the operations given in Definition~\ref{def:model} on
	$k$ fingers, any tournament tree on $n$ keys can be rearranged
	into a path, with keys ordered by the next-request times, 
	at a cost of $O(n \log_{k}(n))$.
\end{lemma}
\begin{proof}
We use the $k$ fingers to simulate $k$-way (external / linked-list based) 
mergesort~\cite[\S5.4.1]{Knuth98a:book},
sorting elements by their target position in the requested permutation.

More specifically, we proceed as follows.
Because the tree is binary, we can find edge separators that
break it into pieces (which are also trees) of size $O(n / k )$.
Thus, with an overhead of $O(n)$ to move the fingers to these
subtrees, plus recursively calling this arrangement procedure
$k$ times on trees of size $O(n / k)$, we can arrange the tree
into $k$ paths attached to the root, each sorted by next request
times.

Then we merge these $k$ paths by advancing $k$ fingers,
one along each path: at each step we simply take the path
whose head contains the lowest next request time.
This once again incurs an overhead of $O(n)$, which means
the overall cost is bounded by the recurrence
\[
C\left( n \right)
\wwrel\leq
k C\left( \frac{n}{k} + 1 \right) + O\left( n \right)
\]
which solves to $O(n \log_{k}n)$.
\end{proof}

This says that any access sequence that is a permutation
is easy to handle for tournament heaps,
and in turn implies a lower bound
on the competitive ratios on the many-fingers case.
\begin{corollary}[Many fingers]
	\label{cor:large-k}
	The competitive ratio of tournament heaps with
	$k$ transient fingers is at least $\Omega(\log k)$.
\end{corollary}

\begin{proof}
	Since we always start at the root, in any tree with $n$
	nodes, there is a node whose cost is at least $\Omega(\log{n})$.
	Incorporating Lemma~\ref{lem:k-lower} gives that the competitive ratio is at least
	\[
	\frac{\Omega\left( \log{n} \right)}{O\left( \log_{k}(n) \right)}
	\wwrel=
	\Omega(\log {k}).
	\]
\end{proof}

\subsection{Putting Things Together}
	
We can now put together the bounds shown above to
prove our main result.

\begin{proof}[\wref{thm:Main}]
	Combining \wref[Corollaries]{cor:small-k}
	and~\ref{cor:large-k} gives that the overall
	competitive ratio is at least
	\[
	\Omega\left(\max\left\{ \frac{\log{n}}{\log{k}}, \log {k} \right\}\right),
	\]
	which is minimized at $\Theta(\sqrt{\log{n}})$
	for $\lg{k} = \sqrt{\lg{n}}$, or
	$k = 2^{\sqrt{\lg{n}}}$.
\end{proof}

As a side remark, we note that this bound is different than
what one would obtain from considering the $k$ fingers as
$k$ servers on a (static) tree.
Specifically, even if we have up to $k = \Theta(\sqrt{n})$
persistent fingers (which remain where they are after each access),
the online cost of accesses is still $\Omega( \log{n})$.
\begin{lemma}[Online lower bound with persistent fingers]
	\label{lem:k-lower}
	In any binary tree with $k$ persistent fingers, are at least $n/2$ nodes whose
	shortest path to a finger has length at least $\lg{n} - \lg{k}  - \lg{3}$.
\end{lemma}

\begin{proof}
	Because the tournament tree is binary,
	for any distance $d$, each finger can reach at most
	$3\cdot 2^{d}$ vertices.
	So the shortest distance to a finger must satisfy
	\(
		3k \cdot 2^{d} \geq n,
	\)
	or
	\(
	d \wwrel\geq \lg{n} - \lg{3} - \lg{k}
	\).~
\end{proof}
Thus, the much lower bound of $O(\log_{k}(n))$ comes from
the ability to rearrange the tree: this is a key
distinction between finger-based searches on (dynamic) trees
and the study of server problems on a static tree (metric)~\cite{ManasseMS90,BansalBMN15,BubeckCLLM18,Lee2018}.

\section{Bucketed Order by Next Request}
\label{sec:loglog}

By \wref{thm:wilber0} and \wref{lem:larger-k-perm} we can  serve any access sequence 
without repetitions, \ie, any permutation, offline in optimal $\Theta(n \log_k (n))$ time.
For general access sequences with repetitions of keys, such a solution seems not at all obvious, 
but we can achieve almost the same result by an algorithm
which buckets elements on the interval until their next request.

\begin{theorem}[Efficient Offline Algorithm]
	\label{thm:large-k-general}
	Given $k$ transient fingers,
	we can perform any sequence of $m$ operations on a
	tournament tree of size $n$ at cost $O(m (\log_{k}n + 2^{\lg^*(n)}))$.
\end{theorem}
This is just \wref{thm:offline-algo} restated, but emphasizing that transient fingers suffice.
This cost is optimal for $k = O(n^{1/(2^{\lg^*(n)})})$, \ie, for for sub-polynomially many fingers.

\begin{proof}[\wref{thm:large-k-general}]
On a high level, our algorithm 
keeps elements in buckets of exponentially increasing sizes and stores those buckets
in a balanced binary tree. Accessing any bucket is then possible in  $\Oh(\log \log n)$ time.
An obvious candidate for defining buckets is the time of the next access to an element.
While this approach is sufficient for the simple array-based data structure of 
\wref{app:elementary-offline-array-algorithm} and~--
in slightly disguised form~-- also constitutes the mechanism behind 
the offline list-update algorithm from~\cite{Munro2000}, 
it seems hard to (efficiently) maintain buckets based on next-access times 
within a binary tree.

Our solution instead uses the concept of \emph{recurrence time},
the time \emph{between} two successive accesses to the same element,
to define buckets.
To allow the maintenance of buckets in amortized constant time per access,
we have to play a second trick: elements inserted into a bucket
are kept separately in two buffers that distinguish next accesses in the ``near future'' 
from accesses in the ``distant future''. That allows to sort buffers once,
without having to deal with insertions into sorted sequences.

To present the details, we fix some notation.
We call $t$ the (global) \emph{time(stamp)} of the access~$a_t$.
We define the \emph{recurrence time} $r(t)$ of the access $a_t$ at time $t$
to be the number of time steps before that same element is requested the \emph{next} time 
in the future:  
\[
		r(t) 
	\wwrel= 
		\min\{t' - t : a_{t'} = a_t, t' > t \}\cup \{\infty\}.
\]
($r(t) = \infty$ if $a_t$ is the last occurrence of that element in the
access sequence.) 
Similarly, define the \emph{next access time} $n(x,t)$ of an access to key $x$
to be the earliest time $t' > t$ when $x$ is requested:
$n(x,t) = \min\{k : a_{t'} = x, t' > t \}\cup \{\infty\}$.
We abbreviate $n(t) = n(a_t,t) = t + r(t)$.

\paragraph{Segments of accesses}
We assume the access sequence has length $m\ge n$.
We divide the accesses into \emph{``segments''} of length $n$ each,
(allowing a potentially incomplete last segment).
At the beginning of each segment, we rearrange the entire tree $T$, so that 
elements that are not requested during this segment are below any elements that will
be requested in the segment.
(We put unused elements ``out of the way''.)
This preprocessing step (a permutation) costs $\Theta(n\log_k (n))$ per segment, 
which is an amortized $\Theta(\log_k(n))$ contribution to 
the costs for any single access.
We can thus focus on the first segment for the rest of this section.
Moreover, we will understand the recurrence times of accesses to be relative to the segment,
\ie, $r(t) = \infty$ if this is the last occurrence of $a_t$ in the first segment.

\paragraph{Buckets for Recurrence Time}
We will use $b=\lceil \lg n\rceil$ \emph{``buckets''} $B_1,\ldots,B_b$ to hold elements, 
grouped by their current recurrence time:
$B_1$ holds elements with recurrence time $1$,
$B_2$ gets recurrence times $2$ and $3$, and in general,
$B_j$ holds all elements with recurrence times $r$ in $[2^{j-1}{}\,..\,2^{j}-1]$.
Each bucket $B_j$ is a subtree, conceptually divided into three parts: 
two input buffers, called \emph{``near-future buffer''} and 
\emph{``far-future buffer'',} and a \emph{``sorted queue''}.

In a tournament tree, these can be represented by a convention like the one
sketched in \wref{fig:buckets}.
Note that in terms of its interface to other buckets, 
each bucket looks like one big binary node: it has pointers for a left and a right child.
We can therefore form a binary tree of buckets; indeed,
we keep the buckets in a balanced binary tree of height 
$\lceil\lg \lceil\lg n\rceil\rceil \le 1 + \lg \lg n$.
Navigating to (the first node of) a bucket thus costs $\Oh(\log \log n)$.

\usetikzlibrary{decorations,decorations.text}

\begin{figure}[tbh]

	\plaincenter{
	\begin{tikzpicture}[
		xscale=1,yscale=1.1
	] 
	\small
		
		\useasboundingbox (-5.5,-7) rectangle (7.25,1.25);

		\node[inner] (x) at (0,0) {} ;
		\node[inner] (x1) at (1.5, -1) {} ;
		\node[inner] (x2) at (0, -2) {} ;
		\node[inner] (x3) at (-1.5, -3) {} ;
		\node[inner] (x4) at (-3, -4) {} ;
		\node[inner sep=5pt] (x5) at (-4.5, -5) {\dots} ;

		\draw (x) to (x1);
		\draw (x1) to (x2);
		\draw (x2) to (x3);
		\draw (x3) to (x4);
		\draw (x4) to (x5);

		\draw[dotted,black!50,line width=1pt] 
			(-0.5, 0.5) .. controls (-0.2, 0.7) and (0.2, 0.7)
			.. (1.8, -0.5) .. controls (2.5, -1) and (1.5, -2)
			.. (0.1, -2.5) .. controls (-3, -4.5) and (-5, -6)
			.. (-5, -5) .. controls (-4.5, -4) and (.8, -1.25)
			.. (0, -1) .. controls (-0.5, -0.5) and (-0.7, 0.2)
			.. cycle;
		;
		\draw[draw=none, line width=1pt,
			postaction={decorate, decoration={text along path, raise=4pt,
			text align={align=center}, text color = black!50, text={sorted queue}}}]
			(-5, -5) .. controls (-4.5, -4) and (.8, -1.25)
			.. (0, -1)
		;

		\node[inner,fill=green!50!black] (x21) at (1.5, -3) {} ;
		\node[inner,fill=green!50!black] (x22) at (3, -4) {} ;
		\node[inner sep=5pt] (x23) at (4.4, -4.9) {\dots} ;
		\draw (x2) to (x21);
		\draw (x21) to (x22);
		\draw (x22) to (x23);

		\draw[rounded corners=15pt,dotted,line width=1pt,color=green!50!black] 
			(1.45,-2.25) -- node[above,sloped,rotate=-2.5] {near-future buffer} 
			(5,-4.7) -- (4.5,-5.5) -- (.7,-3.1) -- cycle
		;

		\node[inner,fill=red] (x31) at (0, -4) {} ;
		\node[inner,fill=red] (x32) at (1.5, -5) {} ;
		\node[inner sep=5pt] (x33) at (2.8, -5.9) {\dots} ;
		\draw (x3) to (x31);
		\draw (x31) to (x32);
		\draw (x32) to (x33);

		\begin{scope}[shift={(-1.5,-1.05)}]
		\draw[rounded corners=15pt,dotted,line width=1pt,color=red!90!black] 
			(1.45,-2.25) -- node[above,sloped,rotate=-2.5] {far-future buffer} 
			(5,-4.7) -- (4.5,-5.5) -- (.7,-3.1) -- cycle
		;
		\end{scope}
		
		\draw[dotted,line width=2pt,color=blue!80!black] 
			(-1, 0) .. controls (-.5, 1.5) and (1.5, 1.2)
		.. (2, 0) .. controls (4, -4) and (5., -3.8)
		.. (5, -4) .. controls (6.5, -5) and (5, -8)
		.. (2, -6.5) .. controls (-2, -4) and (-3, -5.3)
		.. (-4, -6) .. controls (-5, -6.5) and (-6, -5.5)
		.. (-5, -4.5) .. controls (-2, -2) and (-1, 0)
		.. cycle
		; 

		\node at (2, 1) {$B_j$};
				
		\node[circle,dotted,color=blue!80!black,draw,minimum size=50pt, line width=2pt] (l) at (-4.5, -2) {};
		\draw (x) -- (l) node [midway, above, sloped,pos=.7] (xltext) {left} ;
		\node[circle,dotted,color=blue!80!black,draw,minimum size=50pt, line width=2pt] (r) at (6, -3) {};
		\draw (x1) -- (r) node [midway, above, sloped,pos=.65] (rltext) {right};
		
	\end{tikzpicture}}

\caption{%
	Sketch of representation of our (conceptual) buckets in the tournament tree,
	showing the queue and the two buffers.
	To the outside, the buckets look like a binary-tree node and can hence be arranged as a binary tree themselves.
	The shaded nodes each consist of one or two internal nodes and a leaf with the stored key.%
}
\label{fig:buckets}
\end{figure}
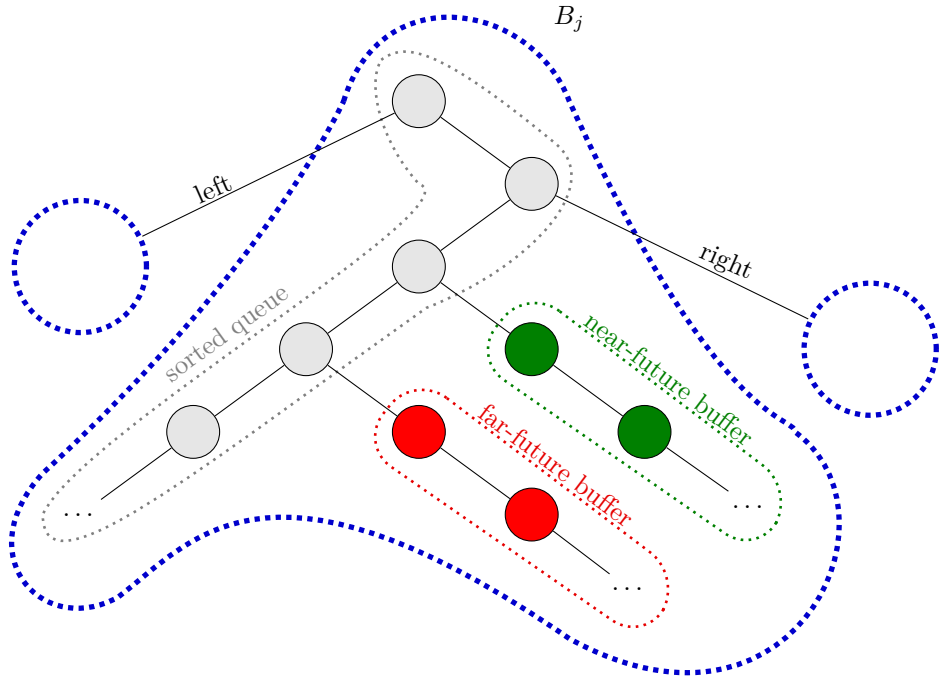

We think of the buffers and the sorted queue of bucket $B_j$ as having a maximal capacity of 
$2^{j-1}$ elements each.
The input buffers are linear lists into which a new node $v$ is inserted 
by making the current buffer
$v$'s child and using $v$ as the new root of the buffer.
The queue is similar, but here elements are only consumed by removing the root of the queue.

Moreover, for each bucket $B_j$, we store an \emph{``expiration time''} 
$e_j$ (for its sorted queue);
this is the time when the sorted queue will resp.\ would become ``invalid''. 
The significance of expiration times will become more clear when we describe 
insertions into buffers below.
The sorted queue can run empty earlier than its expiration time, but we will prove 
(\wref{lem:invariant}) that it always does
so the latest at time $e_j$: 
we do not consume elements past their best-before date.
Initially, the sorted queue is empty and we set $e_j = 2^{j-1} - 1$.

\paragraph{Sorting elements into buckets}
Suppose we serve the access at time $t$ to key $x=a_t$.
$x$ is currently stored in some bucket $B_j$ and will be at the front of $B_j$'s 
sorted queue.
We remove $x$ from the sorted queue of $B_j$
and insert it into a new bucket $B_\ell$, depending on its recurrence time $r(t)$, 
namely so that $r(t) \in [2^{\ell-1}{}\,..\,2^{\ell}-1]$;
(\ie, $\ell$ is the smallest power of two greater than~$r(t)$).

Now, the buffer inside $B_\ell$ into which $x$ is inserted is selected based
on the \emph{absolute time} of the next access to $x$, $t' = n(x,t) = t + r(t)$,
and the expiration time $e_\ell$ of bucket $B_\ell$:
If $t' > e_\ell + 2^{\ell-1}$, then we insert $x$ into the far-future buffer, 
otherwise into the near-future buffer.
The overall procedure is given in \wref{alg:loglog}.

\begin{algorithm}
	\plaincenter{\fbox{~\parbox{.95\linewidth}{%
	\sffamily\small\raggedright\strut
	Repeat for all segments of the access sequence:
	\begin{enumerate}
		\item 
			Sort all elements by first access for this segment 
			and store as sorted queue of $B_0$.
		\item 
			Initialize empty buckets $B_1,\ldots,B_b$. 
		\item 
			For each access $a_t$ in the segment:
		\begin{enumerate}[topsep=0ex]
			\item Let $\rho_1$ be the recurrence time after which $a_t$ occurred,
				\ie, the unique number with $r(a_{t-\rho_1}) = \rho_1$,
				or $\rho_1 \gets \infty$ if element $a_t$ was never accessed before.
			\item	
				\(j \gets \begin{cases*}\lfloor \lg \rho_1\rfloor & if $\rho_1 < \infty$ \\
					0 & otherwise\end{cases*}\)\quad and \\
				\(\ell \gets \begin{cases*}\bigl\lfloor \lg \bigl( r(a_t)\bigr) \bigr\rfloor & if $r(a_t) < \infty$ \\
						0 & otherwise\end{cases*}\).
			\item
				If $t > e_\ell$, call Refresh($B_\ell$).\\
				If $t > e_j$, call Refresh($B_j$).
			\item
				Access the first element, $a_t$, in the sorted queue of bucket $B_j$.
			\item 
				If $n(a_t) \le e_\ell+2^{\ell-1}$, 
				insert $a_t$ into the near-future buffer of $B_\ell$;\\
				otherwise, insert $a_t$ into the far-future buffer of $B_\ell$.
		\end{enumerate}
	\end{enumerate}

	\noindent
	The procedure Refresh($B_j$) refills the sorted queue: 
	\begin{enumerate}
		\item While $t > e_j$ repeat:
		\begin{enumerate}[topsep=0ex]
			\item 
				Make the former near-future buffer of $B_j$ the new queue.\\
				(The old sorted queue is guaranteed to be empty at this stage.)
			\item 
				Make the former far-future buffer of $B_j$ the new near-future buffer.
			\item
				Initialize the far-future buffer of $B_j$ as empty.
			\item
				Set $e_j \gets e_j+2^{j-1}$.
		\end{enumerate}
		\item
			Sort all elements (if any) in the queue of $B_j$ (forming the new \emph{sorted} queue).
	\end{enumerate}

	}}}
	
	\caption{%
		Our doubly-logarithmic offline algorithm for unordered binary trees.
	}
	\label{alg:loglog}
\end{algorithm}

\paragraph{Refreshing Buckets}
When a bucket $B_j$ is about to expire, it is time to ``refresh'' it.
We will opt for a lazy refreshing scheme that allows buckets to remain in an 
expired state, as long as they are not ``touched''.
Here, by touching a bucket, we mean visiting it to access (and extract) an element
or to insert an element into it.
Upon touching a bucket, we check if it has expired, and if so, we refresh it
before continuing.

Refreshing is comprised of the following steps:
We use the current near-future buffer as the new queue, 
and sort its elements by their next access time.
We also make the current far-future buffer the new near-future buffer,
and create a new, empty far-future buffer.
Finally, we advance bucket $B_j$'s expiration time $e_j$ by
$2^{j-1}$, the capacity of $B_j$.

The above steps describe the usual refreshing procedure, but we have
to slightly extend in in general:
A bucket might not have been touched for an arbitrarily long time frame
when the buffers remained empty.
Then, we may have to ``fast-forward'' several $2^{j-1}$ steps before catching 
up with the current time.
Note that in this case, there cannot be any elements in the intermediate queue(s),
so only the last step actually sorts a nonempty queue.
This is reflected in the code in \wref{alg:loglog}.

We initially keep elements in a global sorted queue, 
sorted by first access (for this segment);
we formally call this the zeroth bucket $B_0$.
We build the queue of $B_0$ in the preprocessing step at the beginning of a segment.
We thereby maintain the following invariant:
\begin{lemma}[Bucket invariant]
\label{lem:invariant}
	At any time $t\in[m]$, the bucket $B_j$, $j=1,\ldots,b$, contains exactly the elements 
	whose next access (after $t$) happens after a recurrence time in $[2^{j-1},2^j)$.
	Among those, the sorted queue contains elements with next access at an 
	(absolute) time in $(e_j-2^{j-1}, e_j]$, 
	sorted by next access time,
	the near-future buffer contains elements with next access at a time in $(e_j, e_j+2^{j-1}]$, and
	the far-future buffer those with time in $(e_j+2^{j-1}, e_j+2^j]$.
\end{lemma}
\begin{proof}
The proof is by induction over time $t$.
Initially, buckets are empty and there is no recurrence time before the first access,
so the claim holds.
Let us now assume the claim holds up to time $t-1$. 
At time $t$ there will be a new access, $a_t$ to be served.
If $a_t$ occurs after a recurrence time in $[2^{j-1},2^j)$, we find it in $B_j$'s sorted queue.
(It is vital here that recurrence times do \emph{not} change when we advance $t$, 
whereas the time \emph{until} the next access certainly does.)
Moreover, unless $a_t$ is no longer accessed in this segment, 
it has a new recurrence time $r(t) \in [2^{\ell-1},2^\ell)$, for some $\ell$.
It then has to be (re-)inserted into~$B_\ell$.

Either of these two buckets might have expired, in which case we refresh it.
Assume $B_j$ has expired, \ie, $t>e_j$.
Refresh($B_j$) executes $p \ge 1$ ``promotion rounds'', \ie, $p$ iterations of the white loop,
where $p$ is determined by requiring $t \le e_j+p \cdot 2^{j-1}<t+2^{j-1}$.

\begin{figure}[thbp]
	\plaincenter{
	\begin{tikzpicture}[scale=.9,every node/.style={inner sep=1pt}]
	\small
	\def\r{2pt}
		\fill[black!10] (4.1,-4.5) rectangle ++(2.6,5.8) ;
		\begin{scope}
			\draw[->] (1,0) -- ++(12,0) node[below left] {time (next access)} 
				node[xshift=5em,overlay,anchor=west] {before refresh};
			\draw[ultra thick,dotted, black!50] (6.8,-4.5) -- ++(0,5.8) node[black,above=-2pt] {$t$\rlap{ (now)}};
			\foreach \x/\l in {2/{e_j-2^{j-1}},4/e_j,6/e_j+2^{j-1},8/e_j+2^j} {
				\draw[dotted] (\x,0) -- ++(0,-1) node[] {$\l$};
			}
			\draw[very thick,{Parenthesis[]-Bracket[]},shorten <=1pt] (2,0) -- node[above] {\strut queue} (4,0) ; 
			\draw[very thick,{Parenthesis[]-Bracket[]},shorten <=1pt] (4,0) -- node[above] {\strut near fut.} (6,0) ; 
			\draw[very thick,{Parenthesis[]-Bracket[]},shorten <=1pt] (6,0) -- node[above] {\strut far fut.} (8,0) ; 
			\foreach \x in {6.8,7.5,7.8} {
				\fill (\x,0) circle (\r) ;
			}
			\foreach \x in {8.5,8.8,9.2,10.2,2.5,2.9,3.6} {
				\fill (\x,0) circle (\r) ;
			}
			\draw [<->,shorten >=1pt,shorten <=1pt]
				(6.7,.9) -- node[fill=black!10] {\smaller\itshape no touches!} ++(-2.6,0);
		\end{scope}
		\begin{scope}[shift={(2,-3)}]
			\draw[->] (-1,0) -- ++(12,0)  
				node[xshift=5em,overlay,anchor=west] {after refresh};
			\foreach \x/\l in {4/{e'_j-2^{j-1}},6/e'_j,8/e'_j+2^{j-1},10/e'_j+2^j} {
				\draw[dotted] (\x,1.75) -- ++(0,-1-1.75) node[] {$\l$};
			}
			\draw[very thick,{Parenthesis[]-Bracket[]},shorten <=1pt] (4,0) -- node[above] {\strut queue} (6,0) ; 
			\draw[very thick,{Parenthesis[]-Bracket[]},shorten <=1pt] (6,0) -- node[above] {\strut near fut.} (8,0) ; 
			\draw[very thick,{Parenthesis[]-Bracket[]},shorten <=1pt] (8,0) -- node[above] {\strut far fut.} (10,0) ; 
		\end{scope}
			\foreach \x in {6.8,7.5,7.8} {
				\fill (\x,-3) circle (\r) ;
			}
			\foreach \x in {8.5,8.8,9.2,10.2,2.5,2.9,3.6} {
				\fill (\x,-3) circle (\r) ;
			}
	\end{tikzpicture}
	}	
	\caption{%
		Illustration of a refresh operation with $p=2$ steps.
		The picture shows the ranges of valid next-access times for elements
		in the sorted queue and buffers of $B_j$ before and after the refresh;
		dots indicate times at which $B_j$ is touched.
		Note that $B_j$ cannot have been touched during the gray period for 
		it would have been refreshed earlier then.
	}
	\label{fig:refresh}
\end{figure}
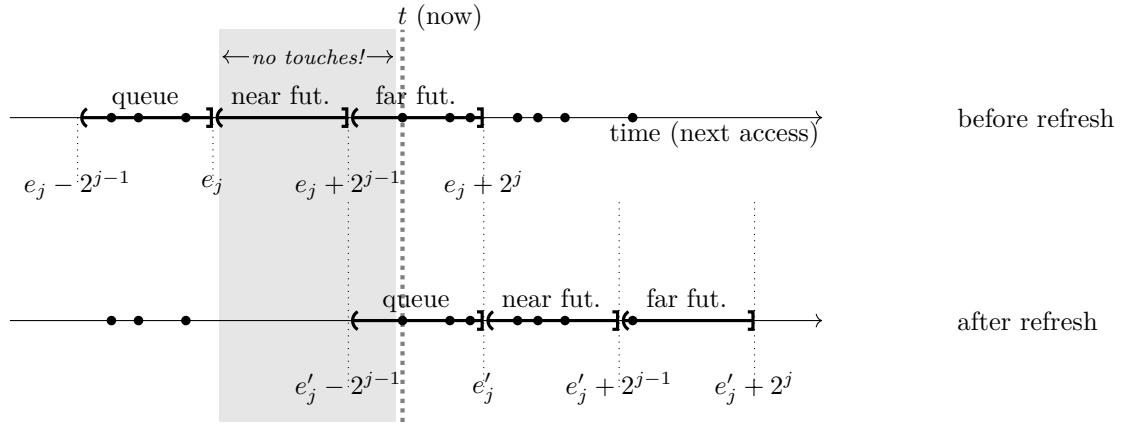

The code in \wref{alg:loglog} only sorts the queue after the last promotion round.
This is sufficient since all temporary queues created in earlier promotion rounds
must be \emph{empty:}
$B_j$ cannot have been touched after its original expiration time $e_j$ since it would have been refreshed then,
so in particular, it is not \emph{accessed} between time $e_j$ and $t$, 
and hence $B_j$ never contains an element with access time $t' \in (e_j,t)$.
Since temporary queues represent time periods before $t$, they are all empty.

Moreover, in each promotion round, the update of $e_j$ exactly undoes the effect of turning 
the near-future buffer into the new queue and the far-future buffer into the new near-future buffer:
the queue always contains elements with next access time $(e_j-2^{j-1}, e_j]$,
the near-future buffer contains elements with next access time in $(e_j, e_j+2^{j-1}]$, and
the far-future buffer those in $(e_j+2^{j-1}, e_j+2^j]$, maintaining the invariant (see \wref{fig:refresh}).

To serve the access $a_t$, we remove it from the sorted queue of $B_j$.
Since that element's next access is no longer at time $t$ (after $a_t$ has been served), 
this reestablishes the invariant.
Moreover, we note that the sorted queue is explicitly sorted upon a refresh and removing the first
element is the only type of update it ever sees, so it remains in sorted order.

Finally, the element $a_t$ is (re-)inserted into $B_\ell$.
If $n(t) > e_\ell + 2^{\ell-1}$, the element is inserted into the far-future buffer.
Since we just refreshed the touched buckets if they were expired, we have $t \le e_\ell$.
With $r(t) < 2^\ell$, this implies $n(t) = t + r(t) < e_\ell + 2^\ell$, 
so $a_t$ fulfills the conditions of the far-future buffer.
If $n(t) \le e_\ell + 2^{\ell-1}$, we insert $a_t$ into the near-future buffer.
It is vital to show that $n(t) > e_\ell$, otherwise it would belong into the sorted queue.
But since $B_\ell$ expires after $2^{\ell-1}$ time steps, we have $t > e_\ell - 2^{\ell-1}$.
With $r(t) \ge 2^{\ell-1}$, the claimed inequality follows: 
$n(t) = t + r(t) > e_\ell - 2^{\ell-1} + 2^{\ell-1} = e_\ell$.
So in all cases, we have reestablished the invariant for time $t$.
\end{proof}
From \wref{lem:invariant}, it immediately follows that our procedure is well-defined:
The element $a_t$ to be accessed will always be at the top of the sorted queue,
readily waiting to be picked up.

\paragraph{Cost analysis}
We divide costs into intra-bucket maintenance and the access costs to reach the buckets in the first place.
To serve one access in a segment, we touch two buckets: 
$B_j$ for retrieving the element from the sorted queue, and $B_\ell$ 
for inserting the element; (possibly $B_j=B_\ell$). 
Inside $B_j$, the requested element is found at the root of the sorted queue,
so we only pay constant extra cost inside $B_j$.
In order to insert the new element into $B_\ell$, 
we need to locally modify a constant number of nodes inside $B_\ell$.
Thus, both operations cause constant additional cost inside the buckets.

If we need to ``refresh'' a bucket, we pay constant overhead to promote the
buffers.
Sorting the near-future buffer costs at most $\Oh(\log_k(n))$ per element
(\wref{lem:larger-k-perm}).
We charge these costs to the next access of each element.
(Every element is sorted only once before it is accessed again, so we 
charge each access at most once for sorting.)

Navigating to one bucket requires navigating down a path of $\Oh(\log \log n)$ nodes each.
We access 2 buckets per access.
The amortized cost of serving one access then consists of $\Oh(1)$ ``intra-bucket maintenance'' costs,
$\Oh(\log \log n)$ to navigate to 2 buckets, and $\Oh(\log_k(n))$ sorting costs.

\paragraph{Hyper Buckets}

Serving an access consists of two conceptually independent steps:
retrieving buckets (source and target) and 
modifying those buckets appropriately.
The intra-bucket part already has optimal amortized cost,
but finding the buckets incurred a $\Oh(\log\log n)$ penalty for navigating
in a binary tree of $b = \lceil \lg n\rceil$ objects,
which is significant for large $k$.
We can improve this by observing that the \textsl{accesses to the buckets are 
themselves an instance of our original problem!}

We (conceptually) contract the buckets into single nodes (cf.\ \wref{fig:buckets})
and assign them ids from $[b]$. Then, executing our doubly-logarithmic algorithm
to serve an access sequence $a_1,\ldots,a_m$ generates a \emph{bucket access sequence}
$a'_1,\ldots,a'_{m'} \in \{B_1,\ldots,B_b\}$ of length $m' \le 2m$, 
but over a universe of only $n' = b=\lceil \lg n\rceil$ different objects.
We recursively apply our offline algorithm on $a'_1,\ldots,a'_{m'}$,
breaking it into segments of $n'$ accesses each, and placing objects into
one of $b' = \lceil \lg b \rceil \le \lg \lg n +1$ buckets.

Iterating this $d\ge 1$ times results in a ``hyper-bucket'' access sequence 
of length $m^{(d)} \le 2^d m$ over $n^{(d)} \le \lg^{(d)} (n) + 1$ different objects,
where $\lg^{(d)}$ denotes the $d$-times iterated logarithm.
Serving this last sequence by keeping the hyper buckets in a static, balanced
tree yields total cost $\Oh( 2^d m \cdot \lg^{(d+1)} (n) )$.
Moreover, we accumulate constant amortized cost over the $d$ levels of recursion
(for maintaining buckets there), giving a total cost of 
$\Oh( \log_k(n) + dm + 2^d m \lg^{(d+1)} (n) )$.
To roughly balance the two factors involving $d$ 
in the third summand, we set $d = \lg^* n$,
yielding total costs in $\Oh(\log_k(n) + m \cdot 2^{\lg^* (n)})$.
Note that sorting costs beyond the topmost level of recursion are
entirely dominated by the sorting costs for topmost buckets and for 
the preprocessing at the beginning of a segment.
This completes the proof of \wref{thm:large-k-general}.
\end{proof}

\section{Conclusion}
\label{sec:conclusion}

In this paper, we investigate models of self-adjusting heaps, specifically,
we study tournament trees.
In the spirit of earlier work on the list-update problem~\cite{MartinezRoura2000,Munro2000},
we point out the importance of the allowed primitive operations for rearranging
the data structure.
We claim that our model of $k$ (transient) fingers, where $k$ is a parameter, 
is a natural choice that allows us to study the continuum between purely local
rearrangement operations and overly powerful global rearrangement.
The influence of different rearrangement primitives is a new facet for
the study of dynamic optimality that is not present for binary search trees,
but enters the game for any type of general heap data structure (not only for tournament trees).
Does the additional freedom in rearrangements give offline algorithms an insurmountable advantage?
Or is there a way to make better use of it also in an online setting (in another model of heaps)?

We show that tournament-tree-based heaps that use 
a decrease-key operation cannot be dynamically optimal, 
totally irrespective of the number of fingers we choose to allow.
Our result invites to try two modifications for getting around the strong separation.
First, one might be willing to sacrifice the ability to modify keys, and
allow only insert and extract-min operations.
\ifarxiv{%
	However, this seems unlikely to defy arguments along the lines given above if suitably adapted.%
}{%
	However, similar arguments as above show that dynamic optimality is out of reach even
	for only extract-min operations.%
}

A second, more promising route is to abandon tournament trees altogether.
We believe that investigating other models of heaps, especially ones
with a less stringent requirement for rearranging paths,
is a highly interesting question;
heap-ordered binary trees and unbounded degree forests immediately come to mind.
We leave their study for future work.

A natural open problem posed by our algorithm is whether the optimal $\Oh(\log_k(n))$
amortized access cost is attainable in the
tournament-tree model with $k$ fingers for any input, 
or if there is an intrinsic cost of insisting on a (binary)-tree-based data structure
(as opposed to random-access memory).

More generally, we also believe that separation between online
and offline algorithms is a far more widespread phenomenon.
The current best running times for offline and online algorithms differ
in a multitude of problems related to dynamic graph data
structures~\cite{Eppstein94,PengSS17:arxiv,HolmDT01,KelnerOSZ13,CohenKMPPRX14}.
Formulating and investigating this separation is an intriguing task
that is significantly beyond the scope of this paper.

\bibliography{references}

\appendix

\section{Offline Algorithms Are Boring With Random Access}
\label{app:elementary-offline-array-algorithm}

As a side comment, we will consider the power of offline algorithms when
they are not restricted by the way in which they store the objects.
We show that we can serve any access sequence offline
in linear time (and thus, optimally) using an array.

\begin{lemma}
	\label{lem:very-very-lazy}
	Any access sequence on $n$ elements can be served by a
	data structure with $n$ persistent pointers in $O(1)$ time per access.
\end{lemma}
\begin{proof}
We divide $A$ into \emph{segments} $A_0,A_1,A_2\ldots$ of $n$ accesses each, \ie,
$A_k = a_{k\cdot n + 1},\ldots,a_{k\cdot n+ n}$ for all $k$ (except possibly an incomplete last round).
At the beginning of each segment,
we create an empty array $R[1..n]$ holding (pointers to) the stored objects.
Iterating over all objects $x$,
we insert $x$ into $R[n(x,k \cdot n) - k\cdot n]$ if $n(x,k\cdot n) - k\cdot n \le n$,
otherwise $x$ remains inactive in this round.
We are now ready to start serving accesses.
The $t\textsuperscript{th}$ request of this segment, $a_{k\cdot n + t} = x$ is found in $R[t]$, 
so we can return $R[t]$ to serve this access.
Moreover, with $j = n(k\cdot n + t)$, the next access time to $x$,
we update the array: if $j - k\cdot n \le n$, we insert a pointer to $x$ into $R[j-k\cdot n]$, 
otherwise $x$ becomes inactive for the rest of this round.
We now continue in the same way with the remaining accesses of $A_k$.
Since we reinsert all elements that are accessed in this round again,
any access can be served by the reference from $R$, which takes constant time each.
The preprocessing at the beginning of a round takes $\Theta(n)$ time, and can thus
be amortized over the next $n$ accesses of the round.
We can thus serve any sequence of accesses in optimal constant amortized time.
\end{proof}

\section{From Keys in Internal Nodes to Keys in Leaves -- And Back}
\label{app:leaf-oriented-BSTs}

In this appendix, we show that from the perspective of dynamic optimality,
standard BSTs and leaf-oriented BSTs are essentially equivalent.

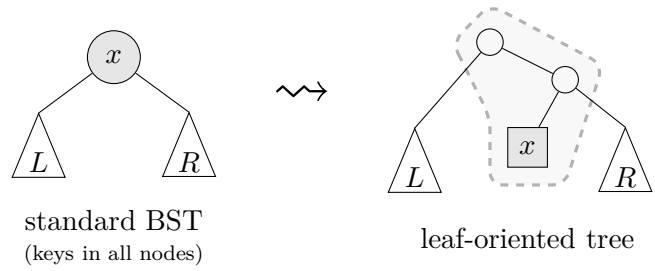
\begin{figure}[tbhp]
	\plaincenter{
	\begin{tikzpicture}
		\node[inner] (x) at (0,0) {$x$} ;
		\node[subtree] (L) at (-1,-1.4) {$L$} ;
		\node[subtree] (R) at (1,-1.4) {$R$} ;
		\draw (x) to ($(L.north)-(0,.35pt)$);
		\draw (x) to ($(R.north)-(0,.35pt)$);
		\node[align=center] at (0,-2.4) {standard BST\\[-.5ex]\textsmaller[2]{(keys in all nodes)}} ;
		
		\node[scale=2] at (2.5,-.5){$\rightsquigarrow$};
		
		\begin{scope}[shift={(5,.2)}]
			\draw[black!30,very thick,dashed,rounded corners=5pt,fill=black!3] 
				(-.5,0) -- (0,-1) -- (0,-1.9) -- ++(1,0) -- ++(.5,1.6) -- ++(-1.5,.8) -- cycle;
			\node[einner] (i1) at (0,0) {};
			\node[einner] (i2) at (1,-.5) {};
			\node[leaf] (lx) at (.5,-1.4) {$x$} ;
			\node[subtree] (lR) at (1.8,-1.8) {$R$} ;
			\node[subtree] (lL) at (-1,-1.8) {$L$} ;
			\draw 
				(i1) to (i2) to (lx)
				(i1) to ($(lL.north)-(0,.35pt)$)
				(i2) to ($(lR.north)-(0,.35pt)$)
			;
		\end{scope}
		\node at (5.5,-2.4) {leaf-oriented tree} ;
	\end{tikzpicture}
	}
	\caption{%
		Transformation from standard BSTs to leaf-oriented trees (for a node with nonempty subtrees).
		When $L$ and/or $R$ are empty, special rules apply: If $x$ is a leaf (both $L$ and $R$ are empty), 
		it is mapped to a leaf with key $x$.
		If $x$ is a unary node, it is mapped to a single internal node with $x$ and its nonempty subtree
		attached (in correct order).
	}
	\label{fig:standard-to-leaf-oriented}
\end{figure}

First of all, we can associate a leaf-oriented tree $T' = \ell(T)$ to any standard BST $T$
by applying the replacement rule shown in \wref{fig:standard-to-leaf-oriented} individually
to all nodes in $T$.
Since the transformation is entirely local (it keeps the structure of $T$ intact within $T'$), 
and replaces each node in $T$ by at most 3 nodes in $T'$,
we immediately obtain the following lemma.
\begin{lemma}
\label{lem:top-subtrees-bst-to-leaf-oriented}
	Let $V_j$ be an arbitrary top-subtree of $T$. Then there is a top-subtree $V_j'$ of $T' = \ell(T)$ 
	containing (the leaves with) all the keys in $V_j$ that satisfies $|V_j'| \le 3 |V_j|$.
	Moreover, let $U$ result from $T$ by replacing $V_j$ by another binary tree over the same nodes.
	Then $U' = \ell(U)$ can be obtained from $T' = \ell(T)$ by only modifying $V_j'$ in $T'$.
\end{lemma}
We can thus simulate a sequence of BST restructuring operations for $T$
in the leaf-oriented model by replacing $V_j'$ with the leaf-oriented tree
corresponding to the new top-subtree in $T$.
By \wref{lem:top-subtrees-bst-to-leaf-oriented}, 
a leaf-oriented BST can thus simulate any BST algorithm with a constant-factor overhead.

The inverse direction is a bit more tricky since not all leaf-oriented BSTs
can be translated back by a purely local transformation:
In a standard BST, we always need a key in the root,
whereas there are leaf-oriented tree with all leaves at depths $\Theta(\log n)$.
However, the following top-down procedure is sufficient for our purposes.
Given a leaf-oriented tree $T'$, we define the standard BST $T=s(T')$ recursively
(see also \wref{fig:leaf-oriented-to-standard}):
If $T'$ is a single leaf, create a single node with that key.
Otherwise, the root of $T'$ is an internal node (without a key).
Let $x$ be the key in the leftmost leaf of the right subtree of the root of $T'$.
$x$ will become the root of $T$ and the subtrees are translated recursively, with
the leftmost leaf in the right subtree removed

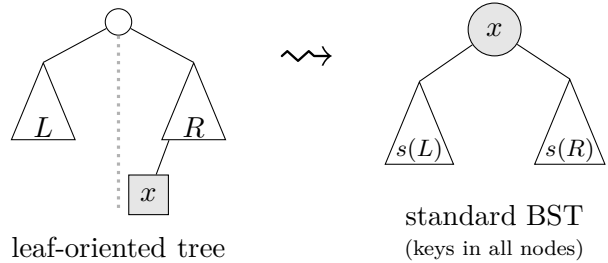
\begin{figure}[tbhp]
	\plaincenter{
	\begin{tikzpicture}
		\node[einner] (r) at (0,0) {} ;
		\node[subtree] (lL) at (-1,-1.4) {\makebox[1.3em]{$L$}} ;
		\node[subtree] (lR) at (1,-1.4) {\makebox[1.3em]{$R$}} ;
		\node[leaf] (x) at (.4,-2.3) {$x$} ;
		\draw (r) to ($(lL.north)-(0,.35pt)$);
		\draw (r) to ($(lR.north)-(0,.35pt)$);
		\draw (x) to ($(lR.left corner)+(3pt,0)$);
		\draw[very thick,dotted,black!30] (r) -- ++(0,-2.5); 
		\node at (0,-3) {leaf-oriented tree} ;

		\node[scale=2] at (2.5,-.5){$\rightsquigarrow$};
		
		\begin{scope}[shift={(5,-.1)}]
			\node[inner] (ox) at (0,0) {$x$};
			\node[subtree] (lR) at (1,-1.6) {\smaller$\!s(R)\!$} ;
			\node[subtree] (lL) at (-1,-1.6) {\smaller$\!s(L)\!$} ;
			\draw 
				(ox) to ($(lL.north)-(0,.35pt)$)
				(ox) to ($(lR.north)-(0,.35pt)$)
			;
		\end{scope}
		\node[align=center] at (5,-2.8) {standard BST\\[-.5ex]\textsmaller[2]{(keys in all nodes)}} ;
	\end{tikzpicture}
	}
	\caption{%
		Transformation from leaf-oriented trees to standard BSTs.
		$x$ is the leftmost leaf in the right subtree of the root,
		which is removed from the recursive call $s(R)$.
	}
	\label{fig:leaf-oriented-to-standard}
\end{figure}

\begin{lemma}
\label{lem:top-subtrees-leaf-oriented-to-bst}
	Let $V_j'$ be any top-subtree in $T'$, containing the set of (leaves with) keys $K_j$.
	Then, there is a top-subtree $V_j$ in $T=s(T')$ that contains all keys $K_j$ and satisfies
	$|V_j| \le |V_j'|$.
	Moreover, let $U'$ result from $T'$ by replacing $V_j'$ by another binary tree over the same nodes.
	Then $U = s(U')$ can be obtained from $T = s(T')$ by only modifying $V_j$ in $T$.
\end{lemma}
\begin{proof}
Consider applying $s$ to $T'$, but stopping the recursion whenever we reach a subtree
that does not contain any node from $V_j'$. The resulting tree $V_j$ will have at most $|V_j'|$ 
nodes since each recursive step removes at least one node of $V_j'$ from further consideration,
and adds one node to $V_j$. Moreover, all keys in $K_j$ are mapped to nodes with these keys,
so they are contained in $V_j$. This proves the first part of the claim.

For the second part, we first observe that removing $V_j'$ from $T'$ disconnects the tree, 
leaving a sequence of subtrees $S'_1,\ldots,S'_k$ behind.
Similarly, removing $V_j$ (as defined above) from $T$, leaves the subtrees $S_1,\ldots,S_k$
behind. The $S_i$ are obtained as follows. Unless all keys in $S_i'$ are smaller than all keys in $K_j$,
remove the leftmost leaf from $S_i'$. Then apply $s$ to the resulting tree.
Note that this leaves some $S_i$ empty when $S_i'$ consisted of a single leaf.
Since we only change $V_j'$ when going from $T'$ to $U'$, 
removing the transformed top-subtree in $U'$ yields the \emph{same} subtrees $S'_1,\ldots,S'_k$.
It follows that $s(U')$ can be obtained by starting at $T$ and only changing $V_j$, as claimed.
\end{proof}

We can hence also simulate any leaf-oriented tree algorithm in a standard BST, 
with at most the same cost.
This constant-overhead bi-simulation result means that 
(a) any constant-competitive
online algorithm for (standard) BSTs also yields such an algorithm for leaf-oriented BSTs,
and vice versa, and 
(b) any lower bounds for one model imply the same lower bound up to constant factors for the other model.

\end{document}